\documentclass[12pt,fleqn]{article}
\usepackage{comment,etex,authblk,amsmath,fullpage,amsfonts,titletoc,amssymb,theorem}

\usepackage[margin=1.25in]{geometry}
\usepackage{setspace}
\AtBeginEnvironment{abstract}{\singlespacing}

\usepackage{graphicx}
\usepackage{booktabs}
\usepackage{caption}
\usepackage{subcaption}

\usepackage{titlesec}
\titleformat{\section}{\normalfont\large\bfseries}{\thesection.}{0.5em}{}
\titleformat{\subsection}{\normalfont\normalsize\bfseries}{\thesubsection.}{0.5em}{}

\usepackage[hidelinks]{hyperref}

\usepackage{pgfplots,dcolumn,ctable,floatpag,ulem}
\usepackage{array, multirow,placeins}

\usepackage{stackengine}

\usepackage{tikz}
\usetikzlibrary{positioning, shapes.geometric}

\titleformat*{\section}{\large\bfseries}
\titleformat*{\subsection}{\normalsize\bfseries}

\hypersetup{colorlinks=true,
linkcolor=blue,
citecolor=blue,
urlcolor=blue}

\usetikzlibrary{pgfplots.groupplots}
\usetikzlibrary{patterns, arrows}
\usetikzlibrary{arrows,positioning,shadows,shapes,calc,fit,backgrounds}
\tikzset{>=stealth}

\reversemarginpar

\floatpagestyle{empty}

\makeatletter
\def\pgfplots@drawtickgridlines@INSTALLCLIP@onorientedsurf#1{}
\makeatother

\usepackage{natbib}
\makeatletter
\renewcommand{\bibsection}{
\begin{center}
\section*{\refname\@mkboth{\MakeUppercase{\refname}}
{\MakeUppercase{\refname}}}
\end{center}
}
\makeatother

\newtheorem{theorem}{Theorem}

\newtheorem{corollary}[theorem]{Corollary}

\newtheorem{definition}{Definition}

\newtheorem{lemma}{Lemma}
\newtheorem{theorem-app}{Theorem}[section]
\newtheorem{lemma-app}[theorem-app]{Lemma}
\newtheorem{proposition-app}[theorem-app]{Proposition}
\theorembodyfont{\upshape}

\newenvironment{proof}[1][\proofname]{
\par\normalfont\trivlist\item[\hskip\labelsep\textbf{#1}.]\ignorespaces}
{\hfill $\square$ %Q.E.D.
\endtrivlist}
\newcommand{\proofname}{Proof}

\newcommand{\gs}{\sigma}

\newcommand{\ga}{\alpha}

\newcommand{\gl}{\lambda}
\newcommand{\eps}{\varepsilon}
\newcommand{\gd}{\delta}

\newcommand{\R}{{\mathbb R}}

\newcommand{\go}{\omega}

\newcommand{\diag}{\operatorname{diag}}

\newcommand{\gS}{{\Sigma}}

\newcommand{\tr}{\operatorname{tr}}

\newcommand{\vertiii}[1]{{\left\vert\kern-0.25ex\left\vert\kern-0.25ex\left\vert #1 
\right\vert\kern-0.25ex\right\vert\kern-0.25ex\right\vert}}

\newcommand{\bK}{{\mathbb K}}

\usepackage{ulem}

\usepackage{bbm}
\usepackage{color}

\usepackage{mathtools}

\newcommand{\be}{\begin{equation}}

\newcommand{\ee}{\end{equation}}
\newcommand{\bea}{\begin{eqnarray}}
\newcommand{\eea}{\end{eqnarray}}
\newcommand{\bee}{\begin{equation*}}
\newcommand{\eee}{\end{equation*}}

\defcitealias{kelly2021virtue}{KMZ}
\defcitealias{jensen2023there}{JKP}
\defcitealias{didisheim2024apt}{DKKM}

\begin{document}

\title{Large and Deep Factor Models
}

\author{}

\author{
Bryan Kelly, Boris Kuznetsov, Semyon Malamud, and Yuan Zhang\thanks{Bryan Kelly is at Yale School of Management, AQR Capital Management, and NBER; \url{www.bryankellyacademic.org}. Boris Kuznetsov is at AQR Capital Management. Semyon Malamud is at the Swiss Finance Institute, EPFL, and CEPR, and is a consultant to AQR. Yuan Zhang is at the Shanghai University of Finance and Economics. Semyon Malamud gratefully acknowledges the financial support of the Swiss Finance Institute and the Swiss National Science Foundation, Grant 100018-228042.
AQR Capital Management is a global investment management firm that may or may not apply similar investment techniques or methods of analysis as described herein. The views expressed here are those of the authors and not necessarily those of AQR. This work was supported by a grant from the Swiss National Supercomputing Center (CSCS) under project ID lp46. We thank Mohammad Pourmohammadi for his excellent comments and suggestions.}  
}
\maketitle
\begin{abstract}\setlength{\parskip}{0pt}
We show that a deep neural network (DNN) trained to construct a stochastic discount factor (SDF) admits an additive decomposition separating nonlinear characteristic discovery from the pricing rule that aggregates them. This decomposition yields a linear factor representation governed by the Portfolio Tangent Kernel (PTK), which summarizes the network’s learned features. In population, the implied SDF converges to a ridge-regularized version of the true SDF, with the degree of regularization determined by spectral complexity. Empirically, using U.S. equity data, the PTK representation delivers economically and statistically significant performance gains, while rising spectral complexity imposes tighter limits on finite-sample pricing.

\noindent\textbf{Keywords:} Asset pricing; stochastic discount factor; neural networks; factor models  
\\
\textbf{JEL classification:} G12; C45; C52
\end{abstract}

%%%%%%%%%%%%%%%%%%%%%%%%%%%%%%%%%%%%%%%%%%%%%%%%%%%%%%
\section{Introduction}
%%%%%%%%%%%%%%%%%%%%%%%%%%%%%%%%%%%%%%%%%%%%%%%%%%%%%%

Modern asset pricing increasingly relies on high-dimensional data and flexible statistical methods. A growing literature shows that deep neural networks (DNNs) can construct stochastic discount factors (SDFs), return forecasts, and portfolio strategies that outperform traditional linear factor models. These empirical successes raise a fundamental but unresolved question: what, exactly, do neural networks learn when trained on financial data?

From an economic perspective, training a neural network implicitly performs two distinct tasks. First, the network learns {\it what to price}: nonlinear features constructed from firm characteristics that summarize economically relevant information. Second, it learns {\it how to price}: a rule that aggregates those features into an SDF. In standard implementations, these two tasks are learned jointly through gradient descent and are therefore tightly entangled. As a result, the network must simultaneously discover informative features and estimate their optimal combination from finite samples. This joint learning problem obscures economic interpretation and can lead to statistical inefficiencies, reinforcing the perception of neural networks as black boxes that are difficult to discipline using asset pricing theory.

This paper shows that DNN-based SDFs admit a sharp additive decomposition that separates feature learning from pricing. For networks trained by gradient descent, the estimated SDF can be written as the sum of two components: a feature-based term that depends only on the nonlinear characteristics learned by the network, and a residual model-dependent term. We show that the feature-based component admits a closed-form representation governed by a new object, the {\it Portfolio Tangent Kernel} (PTK).

Conditional on the learned characteristics, the PTK induces a unique linear pricing rule. The resulting SDF—what we call the {\it PTK-SDF}—is the Markowitz portfolio of a very large collection of characteristic-managed factors. This representation takes the form of a Large Factor Model (LFM). Because this pricing rule is optimal conditional on the learned features, the residual model-dependent component of the original DNN does not improve performance and can be discarded. In this sense, neural networks are best understood as powerful devices for feature discovery, while the PTK provides a cleaner, transparent, and statistically efficient way to price those features.

Empirically, this distinction is decisive. Using U.S.\ equity data, we show that the PTK-SDF systematically outperforms the original DNN-SDF, delivering higher Sharpe ratios and economically significant alphas across asset size groups. Conditioning on learned features while re-optimizing the pricing rule mitigates learning-induced noise that arises when feature discovery and pricing are conducted jointly. This leads to a modular estimation pipeline: a neural network is used solely to learn nonlinear characteristics; these characteristics are mapped into a large set of characteristic-based factors; and the factors are combined using an explicit, regularized portfolio rule. The PTK-SDF therefore isolates—and improves upon—the economically relevant content of the original network.

The closed-form PTK representation also allows us to characterize the statistical limits of DNN-based asset pricing models. Because the number of learned factors is extremely large relative to available time-series observations, classical inference breaks down. Building on recent results in random matrix theory, we show that performance is governed by two forces. The first is {\it alignment}: the extent to which risk premia load on statistically strong directions of factor risk. The second is {\it spectral complexity}: the dispersion of factor risk across principal components, which determines the effective degree of regularization faced by the investor. Higher spectral complexity corresponds to a flatter risk spectrum, in which economically relevant variation is spread across many weak directions. In finite samples, this dispersion amplifies estimation error and limits achievable Sharpe ratios.

We document substantial time variation in the spectral complexity of the PTK. Over the past two decades, it has increased by roughly a factor of six, indicating a pronounced rise in the statistical complexity of the SDF. This increase coincides with the well-documented decline in factor model performance since the early 2000s. At the same time, alignment between risk premia and factor risk has improved, partially offsetting the adverse effects of rising complexity. These findings highlight a fundamental trade-off: feature learning improves economic alignment, but rising spectral complexity imposes limits on what can be learned from finite samples.

An important implication of this framework is that feature learning fundamentally reshapes regularization. We find that factors constructed from learned features exhibit substantially higher spectral complexity than factors constructed from random features, such as those studied by \cite{didisheim2024apt}. Quantitatively, the spectral complexity of learned-feature-based factors is roughly an order of magnitude larger. As a result, these factors are subject to strong {\it implicit} regularization even in the absence of explicit ridge penalties. In particular, the PTK-SDF achieves its best performance with little or no additional shrinkage: despite involving hundreds of thousands of factors, the optimal pricing rule is effectively ridgeless.\footnote{Here, “ridgeless” does not mean unregularized. The PTK-SDF remains strongly disciplined through the spectral complexity of the learned factors, which endogenously shrinks weak risk directions even in the absence of an explicit penalty.}

Beyond statistical performance, the PTK representation restores economic discipline. We show that PTK-based SDFs exhibit substantially stronger alignment with long-horizon and future consumption risk—both in the sense of \cite{ParkerJulliard2005} and its micro-founded extension in \cite{MalamudWangZhang2025}—than either fully trained DNN-SDFs or random-feature benchmarks.

Taken together, our findings provide a unified framework for understanding deep learning in asset pricing. Neural networks learn economically informative features, but pricing those features efficiently requires separating feature discovery from model estimation. The Portfolio Tangent Kernel performs this separation by replacing the original DNN with an explicit large-factor pricing representation. This representation allows both the economic content of learned features and the statistical limits imposed by finite samples to be analyzed using modern asset pricing tools rooted in high-dimensional statistics and random matrix theory.

\section{Related Literature}

A central theme in asset pricing is the tension between sparsity and complexity in the construction of stochastic discount factors (SDFs). While classical factor models emphasize parsimony, a large literature documents that richly parameterized models substantially outperform sparse specifications out of sample \citep[e.g.,][]{cochrane2011presidential,harvey2016and,mclean2016does,Kellyetal2018,guetal2020,hou2020replicating,didisheim2024apt}. See \citep{kelly2023survey} for an overview.

These findings motivate Large Factor Models (LFMs), which construct SDFs from extensive collections of characteristic-based factors. With appropriate shrinkage, LFMs perform well even when based on simple linear characteristics \citep{kozak2020shrinking,kelly2023upsa}. At the same time, approximating the true conditional SDF may require an extremely large number of factors. \citep{didisheim2024apt} show that, in the limit of infinitely many factors, the unconditionally optimal factor portfolio converges to the true conditional SDF, despite the latter being infeasible to estimate directly.

The use of highly over-parameterized factor models raises natural statistical concerns. Estimation error becomes severe when the number of factors exceeds the sample size, and classical inference breaks down. Nevertheless, \citep{didisheim2024apt}, building on the virtue-of-complexity framework of \citep{kelly2022virtue}, show that extreme dimensionality can improve out-of-sample performance. A key insight is that regularization (in the form of effective shrinkage) arises endogenously. \citep{didisheim2024apt}, \citep{kelly2025understanding}, and \citep{chernov2025test} show that effective shrinkage is determined by the spectral properties of covariance matrices, while \citep{chernov2025test} emphasize the role of alignment between factor risk premia and factor risk in determining achievable Sharpe ratios.

These insights motivate new tools for inference. When the number of factors exceeds the sample size, the classical \citep{gibbons1989test} statistic diverges. \citep{chernov2025test} propose a debiased GRS statistic that converges to a ridge-penalized population Sharpe ratio and provides a direct measure of alignment. We adopt this statistic in our empirical analysis.

Most existing Large Factor Models rely on linear factors or random nonlinear transformations. In particular, the model studied by \citep{didisheim2024apt} is equivalent to a shallow but wide neural network with untrained hidden-layer weights. In contrast, much of the machine-learning-based asset pricing literature employs fully trained, multi-layer networks whose economic content and statistical properties are difficult to analyze theoretically \citep[e.g.,][]{chen2019deep,gu2020empirical,fan2022structural}.

This paper bridges these strands by distinguishing between feature learning and pricing in deep learning–based asset pricing models. We study fully trained neural networks of arbitrary depth and width that remain analytically tractable using Neural Tangent Kernel methods \citep{jacot2018neural,yang2020tensor}. We show that, while the DNN learns economically informative nonlinear features, the associated pricing rule need not be optimal. Instead, the learned features admit a closed-form Large Factor Model representation governed by a novel kernel, the Portfolio Tangent Kernel (PTK), which delivers an optimal pricing rule conditional on those features. This perspective clarifies how feature learning reshapes spectral complexity, implicit regularization, and alignment, and connects deep learning models to the emerging theory of high-dimensional asset pricing. 

\section{Parametric SDFs, Kernels, and Neural Networks}

We consider a panel of stocks indexed by $i = 1, \ldots, N_t$, with excess returns
\[
R_{t+1} = \left( R_{i,t+1} \right)_{i=1}^{N_t}
\]
observed at time $t+1$. Each stock $i$ is associated with a $d$-dimensional vector
of characteristics
\[
X_{i,t} = \left( X_{i,t}(k) \right)_{k=1}^d,
\]
and we collect all characteristics in the matrix
\[
X_t = \bigl[ X_t(1), \ldots, X_t(d) \bigr] \in \mathbb{R}^{N_t \times d}.
\]
The true tradable stochastic discount factor (SDF) is given by
\begin{equation}
M_{t+1}\ =\ 1\ -\ \pi_t'R_{t+1}\,,
\end{equation}
where $\pi_t$ denotes the conditionally efficient portfolio
\begin{equation}\label{mark1}
\pi_t\ =\ E_t[R_{t+1}R_{t+1}']^{-1}\,E_t[R_{t+1}]\,,
\end{equation}
which solves the utility maximization problem  
\begin{equation}\label{q-obj1}
\max_{\pi_t}E_t[\pi_t'R_{t+1}\ -\ 0.5 (\pi_t'R_{t+1})^2]\,.
\end{equation}
If $X_t$ contains all relevant conditioning information, then $\pi_t$ in \eqref{mark1} can be written as
\begin{equation}\label{mark2}
\pi_t\ =\ \pi(X_t)\,,
\end{equation}
for some unknown and potentially highly nonlinear function $\pi(X)\,.$ Using the law of iterated expectations, the problem of finding the conditionally optimal portfolio policy can then be rewritten as an {\it unconditional}, non-parametric utility maximization problem:
\begin{equation}\label{non-param}
\min_{\pi(\cdot)}E[(1\,-\,\pi(X_t)'R_{t+1})^2]\ =\ 1\ -\ 2\,\max_{\pi(\cdot)}E\big[U\big(\pi(X_t)'R_{t+1}\big)\big], 
\end{equation}
where $U(x)\ =\ x-0.5 x^2.$ A standard approach to solving \eqref{non-param} for $\pi(X)$ is to posit a sufficiently rich parametric family of functions $f(x;\theta)$ and to estimate the parameter vector $\theta$ using sample analogs of the objective in \eqref{non-param}. Following \citep{didisheim2024apt}, we consider a ridge-penalized version of this objective,
\begin{equation}\label{main-1}
\min_{\theta}L(\theta),\ where\ L(\theta)\ =\ \Bigg(\frac1T\sum_{t=1}^T(1\ -\ f(X_t;\theta)'R_{t+1})^2\ +\ z\,\|\theta\|^2\Bigg)\,,
\end{equation}
where $z$ denotes the ridge penalty and serves as a regularization parameter.

Equation \eqref{main-1} highlights the central tension underlying high-dimensional SDF estimation. Richer function families reduce approximation bias by allowing rich functional forms for $\pi(X)$, but they simultaneously exacerbate estimation error when the dimensionality of $\theta$ is large relative to the sample size. As a result, improvements in in-sample fit need not translate into better out-of-sample SDF performance. Our contribution is to characterize how this bias--variance tradeoff manifests itself in MSRR-based SDF estimation and to show that, in finite samples, limits to learning impose sharp constraints on the achievable Sharpe ratio, even when the true SDF lies within the span of the chosen function family.

A further complication arises from computation. When $\theta$ is high-dimensional, solving \eqref{main-1} exactly is generally infeasible. In practice, estimation relies on iterative optimization algorithms that converge to particular local minima. Recent work in machine learning emphasizes that the choice of optimization algorithm plays a central role in determining out-of-sample performance; see, for example, \citep{zhou2023temperature}. Thus, not only the function family $f(x;\theta)$, but also the manner in which \eqref{main-1} is solved, matters for the properties of the resulting SDF.

The empirical success of deep learning reflects the fact that certain function families—deep neural networks combined with gradient-based optimization—can achieve good generalization even in regimes where the number of parameters exceeds the available sample size; see, for example, \citep{nakkiran2021deep} and \citep{belkin2021fit}. For asset pricing, this raises two distinct questions: what features of the data are learned by these models, and how those features are combined into an SDF. In the sections that follow, we show that addressing these questions separately is key to understanding both the empirical performance and the statistical limits of neural-network-based SDFs.

\subsection{Random Features and Large Factor Models}
%%%%%%%%%%%%%%%%%%%%%%%%%%%%%%%%%%%%%%%%%%%%%%%%%%%%%%%%%%%%%%%%%%%%%
\begin{figure}
\centering
\includegraphics[width=0.9\textwidth]{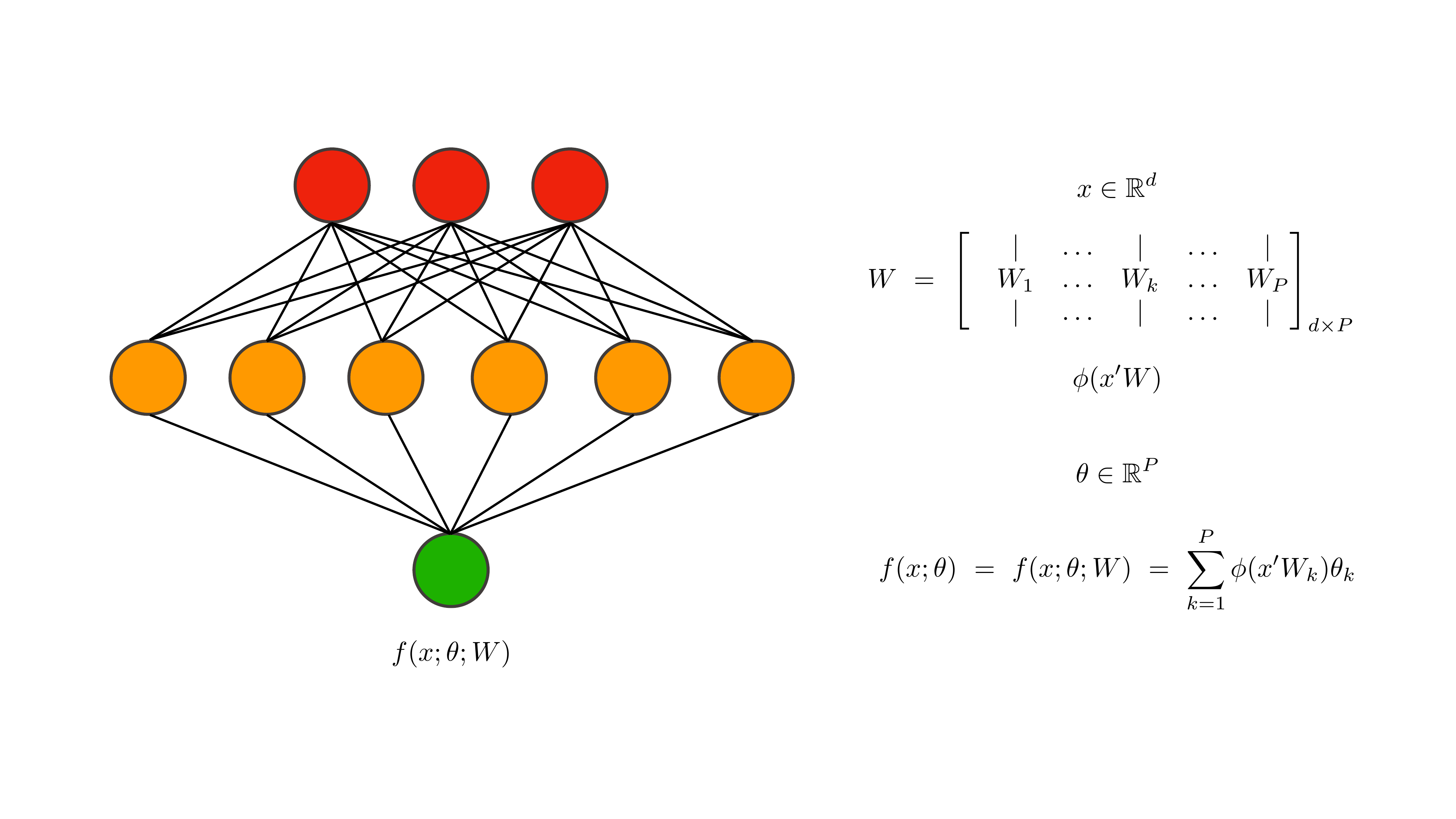}
\caption{The figure above shows $f(x; \theta; W)$, a mathematical representation of a neural network with a single hidden layer, also known as a shallow network. The weights $W \in \R^{d\times P}$ and $\theta \in \R^{P}$ are randomly initialized. $\phi: \R \rightarrow \R$ is an elementwise non-linear activation function and the vector $\phi(x'W)$ is referred to as {\it random features}.}
\label{fig:shallow_layer}
\end{figure}
%%%%%%%%%%%%%%%%%%%%%%%%%%%%%%%%%%%%%%%%%%%%%%%%%%%%%%%%%%%%%%%%%%%%%%%%

Before turning to fully trained deep neural networks, it is useful to consider simpler function families that depend linearly on the parameter vector $\theta\in \R^P.$ These models provide a transparent benchmark in which features are fixed rather than learned, allowing for an explicit characterization of pricing, regularization, and asymptotic behavior. Following \citep{didisheim2024apt}, we consider the family
\begin{equation}\label{dd1}
f(x;\theta)\ =\ f(x;\theta;W)\ =\ \sum_{k=1}^P \phi(x'W_k)\theta_k \,,
\end{equation}
where $x \in \R^d$ denotes the vector of raw characteristics, $W_k\in \R^d$ are randomly drawn weights, and $\phi(\cdot)$ is a nonlinear activation function. The nonlinear functions $\phi(x'W_k)$ are referred to as random features. As explained in \citep{didisheim2024apt}, the specification in \eqref{dd1} is equivalent to a shallow, single-hidden-layer neural network in which the output-layer weights $\theta_k$ are estimated, while the hidden-layer weights $W_k$ are held fixed and not optimized. Figure~\ref{fig:shallow_layer} provides a schematic illustration.

The linearity of the function family in \eqref{dd1} implies that the optimization problem in \eqref{main-1} admits a unique closed-form solution. Moreover, this solution can be expressed in terms of characteristics-managed portfolios, or factors. Specifically, defining the vector of factor returns $F_{t+1}\ =\ (F_{k,t+1})_{k=1}^P$ with
\begin{equation}
F_{k,t+1}\ =\ \sum_{i=1}^{N_t}\underbrace{\phi(X_{i,t}',W_k)}_{k\text{-}th\ random\ feature}\,R_{i,t+1}\,,
\end{equation}
we can rewrite the SDF as a factor portfolio,
\begin{equation}\label{sdf=fac}
\sum_{i=1}^{N_t} f(X_{i,t};\theta)R_{i,t+1}\ =\ \theta'F_{t+1}\,.
\end{equation}
Substituting \eqref{sdf=fac} into \eqref{main-1} yields
\begin{equation}\label{msrr-th0}
\theta(z)\ =\ \arg\min_{\theta}L(\theta)\ =\ \left(zI+T^{-1}\sum_{t=1}^T F_tF_t'\right)^{-1}T^{-1}\sum_{t=1}^T F_t\,.
\end{equation}
Thus, within the linear family \eqref{dd1}, the estimated SDF admits a clear and familiar interpretation as the mean--variance efficient portfolio of a (potentially very large) collection of factors.

Standard intuition based on the Arbitrage Pricing Theory of \citep{Ross1976} suggests that the number of factors required to span the SDF should be small. However, \citep{didisheim2024apt} show both theoretically and empirically that this intuition fails in high-dimensional settings. They establish the {\it virtue of complexity}, whereby out-of-sample performance is monotonically increasing in the number of factors $P$, even in the limit as $P\to\infty$. This result naturally raises the question of the limiting behavior of the estimated SDF as the number of random features grows. Addressing this question requires introducing tools from kernel theory.

A positive definite kernel on $\R^D$ is a function $K(x,\tilde x):\R^D\times \R^D\to \R$ such that, for any collection of vectors $x_i\in \R^D,i=1,\cdots,n,$ and any coefficient vector $\ga\in \R^n,$
\begin{equation}
\sum_{i_1=1}^n\sum_{i_2=1}^n\ga_{i_1}\ga_{i_2}K(x_{i_1},x_{i_2})\ \ge\ 0\,.
\end{equation}
Equivalently, the kernel matrix $\bar K\ =\ (K(x_{i_1},x_{i_2}))_{i_1,i_2=1}^n$ is positive semidefinite. Given a collection of nonlinear functions $\phi_k(x),\ k=1,\cdots,P,$ commonly referred to as features, one can define the associated kernel
\begin{equation}\label{ker=fea}
K(x,\tilde x)\ =\ \frac1P\sum_{k=1}^P \phi_k(x)\phi_k(\tilde x)\,.
\end{equation}
By direct calculation, $K$ is positive semidefinite.\footnote{Indeed, $\sum_{i_1=1}^n\sum_{i_2=1}^n\ga_{i_1}\ga_{i_2}K(x_{i_1},x_{i_2})=\frac1P\sum_k (\sum_i \phi_k(x_i)\ga_i)^2 \ge\ 0.$} Conversely, any positive definite kernel can be approximated by a sufficiently rich collection of nonlinear features: By Mercer’s theorem (see, e.g., \citealp{bartlett2003reproducing,rahimi2007random}), any positive definite kernel $K(x,\tilde x)$ admits the representation
\begin{equation}\label{bochner}
K(x,\tilde x)\ =\ \int \phi(x;\go)\,\phi(\tilde x;\go)\,p(d\go)\,,
\end{equation}
for some probability distribution $p(d\go)$. Sampling $W_k$ from $p(d\go)$ yields the Monte Carlo approximation
\begin{equation}\label{mercer}
K(x,\tilde x)\ =\ \lim_{P\to\infty}\frac1P \sum_{k=1}^P \phi(x;W_k)\,\phi(\tilde x;W_k)\,.
\end{equation}
This representation implies that, as $P\to\infty$, the SDF based on the optimal $\theta$ in \eqref{msrr-th0} converges to a well-defined limit determined solely by the kernel $K$. To characterize this limit in an asset-pricing setting, we introduce the following definition.

\begin{definition}[The Portfolio Kernel]\label{def:pk}
Let
\begin{equation}\label{market-state}
Y_t\ =\ (R_{t+1},X_t)\ \in\ \R^{N_t(d+1)}
\end{equation}
denote the market state, encompassing all stock characteristics and all stock returns. 
Given any two market states $Y=(R,X),\ \tilde Y=(\tilde R,\tilde X),$ with $R\in \R^N,\ \tilde R\in \R^{\tilde N},$ we let
\begin{equation}
K(X,\tilde X)\ =\ (K(X_i,\tilde X_j))\ \in\ \R^{N\times \tilde N}\,,
\end{equation}
and define the {\it Portfolio Kernel} as
\begin{equation}
\bK(Y;\tilde Y)\ =\ \underbrace{R'}_{1\times N}K(X,\tilde X)\underbrace{\tilde R}_{\tilde N\times 1}\,. 
\end{equation}
\end{definition}

The Portfolio Kernel operates on market states and produces a kernel-based similarity measure $\bK(Y;\tilde Y)$. Importantly, this similarity measure is itself a portfolio return, as it aggregates asset returns using weights determined by the similarity of characteristics across market states.

We denote the in-sample data by
\begin{equation}
Y_{IS}\ =\ (R_{t+1}, X_t)_{t=1}^T \in \R^{T\times N_t (d+1)},
\end{equation}
and define the corresponding in-sample kernel matrix
\begin{equation}\label{k-is}
\bK_{IS}\ =\ (\bK(Y_{t_1};Y_{t_2}))_{t_1,t_2=1}^T\ \in\ \R^{T\times T}\,,
\end{equation}
which measures similarity across in-sample market states.

Let $K$ be a positive definite kernel admitting the representation \eqref{bochner}. Let $\{W_k\}_{k=1}^P$ be independently sampled from $p(d\go)$ in \eqref{bochner}, and define the factor vector $F_{t+1}=(F_{k,t+1})_{k=1}^P$ by
\begin{equation}
F_{k,t+1}\ =\ \frac1{P^{1/2}}\sum_{i=1}^{N_t}\phi(X_{i,t};W_k)\,R_{i,t+1}\,.
\end{equation}
Equivalently,
\begin{equation}
F_{t+1}\ =\ \frac1{P^{1/2}}\phi(X_t;W)R_{t+1}\ \in\ \R^P\,.
\end{equation}
Let
\begin{equation}\label{msrr-th}
\theta(z;T)\ =\ \left(zI+T^{-1}\sum_{t=1}^T F_tF_t'\right)^{-1}T^{-1}\sum_{t=1}^T F_t
\end{equation}
denote the ridge-penalized efficient factor portfolio, and let $\theta'F_{t+1}$ be the corresponding SDF. By direct calculation,
\begin{equation}\label{lim-1}
\begin{aligned}
F_{t+1}'F_{\tau+1}
&=\ \frac1P R_{t+1}'\phi(X_t;W)\phi(X_\tau;W)'R_{\tau+1}\\
&=\ R_{t+1}'\left(\frac1P\sum_{k=1}^P \phi(X_t;W_k)\phi(X_\tau;W_k)'\right)R_{\tau+1}\\
&\ \underbrace{\to}_{\eqref{mercer}}\ R_{t+1}'K(X_t,X_\tau)R_{\tau+1}\ =\ \bK(Y_t,Y_\tau)\,.
\end{aligned}
\end{equation}
Denoting by $F=(F_t)_{t=1}^T\in \R^{P\times T}$ the matrix of in-sample factor returns, it follows that
\begin{equation}\label{lim-2}
F'F\ =\ (F_{t_1}'F_{t_2})_{t_1,t_2=1}^T\ \approx\ (\bK(Y_{t_1},Y_{t_2}))_{t_1,t_2=1}^T\ =\ \bK_{IS}\,.
\end{equation}
We will require the following technical lemma.

\begin{lemma}\label{ch-o-lem}
Let $F\in \R^{P\times T}$ denote the matrix of in-sample factor returns, and let ${\bf 1}=(1,\cdots,1)\in \R^T$ be the vector of ones. Then,
\begin{equation}
\theta\ =\ T^{-1}(zI+T^{-1}FF')^{-1}F{\bf 1}\ =\ T^{-1}F(zI+T^{-1}F'F)^{-1}{\bf 1}\,.
\end{equation}
\end{lemma}

Lemma \ref{ch-o-lem} follows directly from the regression objective \eqref{main-1}: $\theta$ is the coefficient vector obtained by regressing ${\bf 1}$ on factor returns. Consequently,
\begin{equation}
\begin{aligned}
F_{T+1}'\theta
&=\ T^{-1}\underbrace{F_{T+1}'F}_{\to\ \bK(Y_{T+1},Y_{IS})}
\left(zI+T^{-1}\underbrace{F'F}_{\to\ \bK_{IS}}\right)^{-1}{\bf 1}\,,
\end{aligned}
\end{equation}
and we obtain the following result.

\begin{theorem}[LFM-SDF]\label{th:lfm}
The Large Factor Model (LFM) SDF is given by
\begin{equation}
R^{LFM}(z)_{T+1}\ =\ \lim_{P\to\infty}\theta'F_{T+1}\ =\ \bK(Y_{T+1},Y_{IS})\xi\ =\ \sum_{t=1}^T \underbrace{\bK(Y_{T+1},Y_t)}_{attention\ to\ t}\ \underbrace{\xi_t}_{optimal\ weight}\,,
\end{equation}
where
\begin{equation}
\xi\ =\ \frac1T (zI+T^{-1}\bK_{IS})^{-1}{\bf 1}\,.
\end{equation}
\end{theorem}
Theorem \ref{th:lfm} characterizes the asymptotic behavior of the random-feature-based Large Factor Model studied in \citep{didisheim2024apt}. In the limit of infinitely many random features, the estimated SDF converges to a well-defined object governed entirely by the kernel $K$.

The LFM-SDF compares the current market state $Y_{T+1}$ with past realizations $Y_t,\ t=1,\cdots,T,$ using the Portfolio Kernel $\bK$. Periods that are more similar to the present receive higher weight, while less similar periods receive lower weight. These similarity scores are then aggregated in a mean-variance efficient manner through the weight vector $\xi\in \R^T$, which can be interpreted as optimal time fixed effects.

Although the kernel representation emphasizes aggregation over time, the LFM-SDF also admits a conventional portfolio representation in terms of characteristics-managed factors.

\begin{corollary}
We have
\begin{equation}
\lim_{P\to\infty}\theta'F_{T+1}\ =\ R_{T+1}'f^\bK(X_t)\,,
\end{equation}
where
\begin{equation}
f^\bK(x)\ =\ \sum_{t=1}^T\Big(\sum_{j=1}^{N_t} R_{j,t+1} K(x, X_{j,t})\Big)\,\xi_t\,.
\end{equation}
\end{corollary}
The representation of $f^\bK(x)$ provides a transparent interpretation of how the LFM extracts pricing information from the joint distribution of returns and characteristics. For any out-of-sample characteristic vector $x$, the model evaluates its similarity to past characteristics through the kernel $K(x,X_{j,t})$ and forms the similarity-weighted performance measure
\begin{equation}
Perf_t(x)\ =\ \sum_{j=1}^{N_t} R_{j,t+1} K(x, X_{j,t})\,.
\end{equation}
The final portfolio weight is obtained by optimally aggregating these performance measures across time,
\begin{equation}
f^\bK(x)\ =\ \sum_{t=1}^T Perf_t(x)\,\xi_t\,.
\end{equation}
In summary, once a kernel $K$ is specified, the analysis delivers an explicit closed-form SDF that efficiently aggregates an effectively infinite number of characteristics-based factors. This property motivates the term Large Factor Model. The analysis also highlights a limitation: the kernel $K$ is exogenously chosen rather than learned. In the next section, we show how fully trained neural networks endogenously learn such kernels through gradient-based optimization, leading to the Portfolio Tangent Kernel.

\section{The Portfolio Tangent Kernel}

We consider a parametric function family $f(X;\theta):\ \R^{N\times d}\to\R^N$ mapping
characteristics $X\in \R^{N\times d}$ to portfolio weights. For example, an
own-characteristics neural network $g(X_{i,t};\theta)$ studied in
\cite{gu2020empirical, chen2019deep} implies
$f(X_t;\theta)\ =\ (g(X_{i,t};\theta))_{i=1}^N$, while a transformer architecture as in
\cite{kelly2025artificial} yields a mapping that cannot be decoupled across stocks
because each coordinate depends on the characteristics of all assets. Throughout, we refer to the function family as a DNN, although it may originate from other families,
such as panel trees \citep{cong2025growing}.

\subsection{Optimization, Initialization, and Gradient Flow}

In practice, DNNs are trained by randomly initializing the parameter vector
$\theta=\theta(0)$ and then updating it using gradient-based optimization. This
initialization is economically relevant: in over-parameterized settings, many
parameter vectors achieve identical in-sample fit, so the learned SDF depends not only
on the objective function but also on the optimization path induced by the algorithm.

Letting $\eta>0$ denote the learning rate,  gradient descent updates the parameters according to
\begin{equation}\label{gr-d}
\theta(s+ds)\ =\ \theta(s) - \eta\, ds\, \nabla_{\theta} L(\theta(s))\,,
\end{equation}
where $s\ge 0$ indexes (rescaled) training time, or equivalently, the cumulative number of
gradient descent steps. In the limit of an infinitesimal step size, this recursion
converges to the gradient flow
\begin{equation}\label{n-ode}
\frac{d\theta(s)}{ds}\ =\ -\,\eta\, \nabla_{\theta} L(\theta(s))\,.
\end{equation}
The path $\{\theta(s)\}_{s\ge 0}$ therefore defines a sequence of candidate pricing rules
generated by the learning algorithm.

While the dynamics of the high-dimensional parameter vector $\theta(s)$ are generally
intractable, the evolution of the portfolio return generated by the network admits a
simpler characterization. Applying the chain rule yields
\begin{equation}\label{neural-dyn}
\frac{d}{ds}(R'f(x;\theta(s)))\ =\ -\,\eta\, R'\nabla_\theta f(x;\theta(s))
\nabla_{\theta} L(\theta(s))\,.
\end{equation}
Equation \eqref{neural-dyn} shows that gradient descent does not update the SDF directly.
Instead, learning operates through the gradient feature portfolio returns  
$R'\nabla_\theta f(x;\theta(s))$, which determine how pricing errors translate into
changes in portfolio weights. From an asset pricing perspective, these gradient feature portfolios define a large collection of characteristics-based payoffs that are implicitly traded
during training. This observation motivates a kernel representation that measures
similarity across market states through the interaction of returns with gradient
features.

\subsection{The Portfolio Tangent Kernel}

\begin{definition}[The Portfolio Tangent Kernel (PTK)]\label{ptk:def}
Let $R\in \R^N,\ \tilde R\in \R^{\tilde N}$ be two return vectors, and let
$X\in \R^{N\times P},\ \tilde X\in \R^{\tilde N\times P}$ denote their characteristics.
Define
\begin{equation}
\bK((R,X);(\tilde R,\tilde X);\theta)\ =\
\underbrace{R'}_{1\times N}
\underbrace{\nabla_\theta f(X;\theta)}_{N\times P}
\underbrace{\nabla_\theta f(\tilde X;\theta)'}_{P\times \tilde N}
\underbrace{\tilde R}_{\tilde N\times 1}
\end{equation}
to be the tangent kernel associated with the portfolio mapping $f(X;\theta)$. We refer
to this object as the {\it Portfolio Tangent Kernel} (PTK).
\end{definition}

The PTK is a pricing-relevant kernel defined on market states. It measures similarity
between states by aggregating returns using weights determined by the sensitivity of
portfolio payoffs to parameter perturbations. By construction, the PTK is a Portfolio
Kernel in the sense of Definition~\ref{def:pk}, with the underlying kernel given by the
Neural Tangent Kernel (NTK; \citealp{jacot2018neural}),\footnote{When the NTK is computed using the DNN after training, it is commonly referred to as the after-kernel. See, \cite{long2021properties, wei2022more, schwab2025training}.}
\begin{equation}\label{ntk-def}
K(X,\tilde X;\theta)\ =\ \nabla_\theta f(X;\theta)\,\nabla_\theta f(\tilde X;\theta)'\,.
\end{equation}
Accordingly,
\begin{equation}
\bK((R,X);(\tilde R,\tilde X);\theta)\ =\ R'K(X,\tilde X;\theta)\tilde R\,.
\end{equation}
While the NTK captures similarity across characteristics, the PTK lifts this structure
to traded returns and therefore governs pricing dynamics rather than prediction alone.

When $\theta$ is high-dimensional, the PTK aggregates a large collection of
characteristics-managed factors,
\begin{equation}\label{factor-ret}
F_{t+1}\ =\ (F_{k,t+1})_{k=1}^P,\qquad
F_{k,t+1}\ =\ R_{t+1}'\nabla_{\theta_k}f(X_t;\theta)\,.
\end{equation}
For any two market states $Y_t,Y_\tau$,
\begin{equation}\label{fac-sim-ntk}
\bK(Y_t;Y_\tau;\theta)\ =\ F_{t+1}'F_{\tau+1}\,,
\end{equation}
so that the PTK measures similarity across time in terms of the realized payoffs of
these learned factors.

\subsection{Feature Learning and PTK}

Equation \eqref{neural-dyn} implies that, under gradient flow, the portfolio return generated by the DNN evolves according to
\begin{equation}\label{neural-dyn-factors}
\frac{d}{ds}\big(R'f(X;\theta(s))\big)
\ =\ -\,\eta\,F_{t+1}'\nabla_{\theta} L(\theta(s))\,.
\end{equation}
Specializing to the MSRR objective yields the following result.

\begin{theorem}[MSRR Gradient Flow and the PTK]\label{thm:ptk} Suppose that 
\begin{equation}\label{main-11}
L(\theta)\ =\ \frac1T\sum_{t=1}^T(1\ -\ f(X_t;\theta)'R_{t+1})^2\,.
\end{equation}
Under gradient flow, 
\begin{equation}\label{ptk-dyn}
\frac{d}{ds}\big(R'f(X;\theta(s))\big)
\ =\ -\,\eta \, T^{-1}\underbrace{\bK(Y;Y_{IS};\theta(s))}_{PTK}\,
\underbrace{\big({\bf 1}_T - R_{IS}'f(X_{IS};\theta(s))\big)}_{in\text{-}sample\ error}\,,
\end{equation}
where $Y_{IS}=(R_{IS},X_{IS})$ is in-sample (IS) data, and 
\begin{equation}
\bK(Y;Y_{IS};\theta)\ =\ \big(\bK(Y;Y_t;\theta)\big)_{t=1}^T\in \R^{1\times T}\,.
\end{equation}
\end{theorem}
Theorem~\ref{thm:ptk} shows that gradient descent updates portfolio returns by comparing the current market state to past in-sample states through the PTK. Learning therefore
operates through two distinct channels:
\begin{itemize}
\item {\it Model learning}, captured by the evolution of $f(X;\theta(s))$, which
determines the portfolio weight function; and
\item {\it Feature learning}, captured by the evolution of the PTK
$\bK(Y;Y_{IS};\theta(s))$, reflecting learning of the gradient features
$\nabla_\theta f(X;\theta)$ used to construct the traded factors.
\end{itemize}
After sufficiently many gradient descent steps, these dynamics stabilize. This
occurs, for example, when $\theta(s)$ approaches a stationary point or when the network
is sufficiently wide \citep{jacot2018neural}. Once feature learning is complete, the PTK
ceases to evolve, and the dynamics admit a closed-form solution.\footnote{The proof of this result follows directly from Theorem \ref{thm:ptk}.}

\begin{theorem}[Closed-Form Stable-PTK DNN-SDFs]\label{main-th-ptk}
Fix a large $s_*$ and define
\begin{equation}
\bK^*\ =\ \bK(\cdot;\cdot;\theta(s_*))\,.
\end{equation}
Suppose that, after $s_*$ gradient descent steps, the PTK stabilizes in the sense that
\begin{equation}
\|\bK(Y;Y_{IS};\theta(s))-\bK^*(Y;Y_{IS})\|
+\|\bK(Y_{IS};Y_{IS};\theta(s))-\bK^*(Y_{IS};Y_{IS})\|
<\eps
\end{equation}
for all $s>s_*$. Suppose further that
\begin{equation}
\bK^*_{IS}\ \equiv\ \bK^*(Y_{IS};Y_{IS})
\end{equation}
is non-degenerate. Then, for all $s>s_*$,
\begin{equation}\label{decomposition-dnn}
R'f(X;\theta(s))\ =\ 
\underbrace{R'f(X;\theta(s_*))}_{learned\ model}
\ +\ 
\underbrace{\frac1T\bK^*(Y,Y_{IS})\,\xi(s)}_{learned\ features}
\ +\ O(\eps)\,,
\end{equation}
where
\begin{equation}
\xi(s)\ =\ 
\Big(\tfrac1T\bK^*_{IS}\Big)^{-1}
\Big(I-e^{-\eta (s-s_*) \tfrac1T\bK^*_{IS}}\Big)
\big({\bf 1}-R_{IS}'f(X_{IS},\theta(s_*))\big)\,.
\end{equation}
\end{theorem}
Theorem~\ref{main-th-ptk} delivers the central decomposition of the paper. Once the PTK
has stabilized, gradient descent no longer alters the feature representation learned by
the network. Instead, subsequent learning acts only through the aggregation of a fixed
collection of gradient-based factors. In this regime, the DNN supplies the features,
while pricing is governed entirely by the PTK.

This result makes precise the separation between {\it what is learned} and {\it how it
is priced}. The term labeled {\it learned model} captures the residual pricing rule
associated with the network at the stabilization point $s_*$. The second term—the
PTK-SDF—corresponds to an explicit factor pricing model that optimally aggregates the
learned features. Because this aggregation is efficient conditional on the features, the
learned-model component does not improve pricing performance and can be discarded
without loss.

The structure of $\xi(s)$ reveals that the PTK-SDF implements spectral shrinkage closely
related to ridge regression. To see this, consider the spectral decomposition of the
in-sample kernel matrix, $\bK^*_{IS}\ =\ UDU',\  D=\diag(\gl_i)_{i=1}^T\,.$ Then,
\begin{equation}
\begin{aligned}
&\Big(\tfrac1T\bK^*_{IS}\Big)^{-1}
\big(I-e^{-\tfrac1T\eta \bK^*_{IS}(s-s_*)}\big)
\ =\ U f_{GD}(D)U'\,,\\
&(zI+\tfrac1T\bK^*_{IS})^{-1}
\ =\ U f_{ridge}(D)U'\,,
\end{aligned}
\end{equation}
where
\begin{equation}
\begin{aligned}
&f_{GD}(\gl)\ =\ \gl^{-1}\big(1-e^{-\eta (s-s_*) \gl}\big)\,,\\
&f_{ridge}(\gl)\ =\ \gl^{-1}\frac{1}{1+\gl^{-1}z}\,,
\end{aligned}
\end{equation}
are the spectral filters induced by gradient descent (GD) and ridge regularization,
respectively. They both regularize the factor covariance matrix $\bK^*_{IS}$ by downweighting directions associated 
with small eigenvalues. The key distinction is that, under gradient descent, shrinkage is generated endogenously by the learning dynamics
rather than imposed ex ante through an explicit penalty. The effect of this shrinkage is most pronounced for small eigenvalues $\gl$, for which
\begin{equation}
\begin{aligned}
&f_{GD}(\gl)\ =\ \eta (s-s_*)\ +\ O(\gl)\,,\\
&f_{ridge}(\gl)\ =\ z^{-1}\ +\ O(\gl)\,.
\end{aligned}
\end{equation}
Thus, gradient descent induces an {\it effective} ridge penalty of order
$1/(\eta (s-s_*))$.

If the kernel matrix $\bK^*_{IS}$ is strictly positive definite, then as $s\to\infty$, the exponential term vanishes, and the PTK-SDF converges to the ridgeless kernel
regression
\begin{equation}
\bK^*(Y,Y_{IS})\,(\bK^*_{IS})^{-1}{\bf 1}\,.
\end{equation}
As we show in the next
section, although this estimator interpolates the in-sample data, it remains disciplined out of sample due to implicit spectral shrinkage induced by the PTK. The strength of this implicit regularization is governed by the spectral complexity of the PTK and plays a central role in shaping out-of-sample SDF performance.

\section{Alignment and the Projection Theorem}\label{sec:alignment}

Recent work emphasizes that limits to learning—the gap between the population-optimal
SDF and the model learned from finite samples \citep{chen2025limits}—are governed by two
fundamental properties of high-dimensional pricing models
\citep{didisheim2024apt, kelly2025understanding}. The first is the spectral distribution
of factor risk, which determines the severity of statistical shrinkage. The second is
{\it alignment}: the extent to which economically relevant risk premia load on
statistically strong directions of the factor space.

Random features are generically poorly aligned with the true SDF. Learned features, by contrast, may become substantially better aligned through training, allowing neural
networks to outperform despite operating in extremely high dimensions.

We formalize these ideas using tools from Random Matrix Theory (RMT). Let
$F\in \R^{T\times P}$ denote the matrix of factor returns, and define
\[
\bar F\ =\ \frac1T \sum_{t=1}^T F_t,\qquad
\hat\Sigma\ =\ \frac1T \sum_{t=1}^T F_tF_t'-\bar F\bar F'\,.
\]
Define the effective penalty
\begin{equation}\label{dfn:zstar}
\hat Z_*(z)\ =\ \frac{1}{T^{-1}\tr\!\big((zI+FF'/T)^{-1}\big)}\,.
\end{equation}

\begin{theorem}\label{thm:master}
Suppose $F_{t+1}=\mu_P+\gS_P^{1/2}X_t$, where $X_t$ has i.i.d.\ components with $E[X_{i,t}^{8+\eps}]<\infty$. Then,
\begin{equation}\label{aym-main}
(zI+\hat\gS)^{-1}\bar F\ -\ \frac{\hat Z_*(z)}{z}(\hat Z_*(z)I+\gS_P)^{-1}\mu_P\ \to\ 0
\end{equation}
in the weak sense as $P,T\to\infty$ with $P/T\to c>0$.
\end{theorem}

Theorem \ref{thm:master} implies that all estimation noise collapses asymptotically to the scalar $\hat Z_*(z)$. When $P>T$,
$\hat Z_*(0)>0$, implying nontrivial shrinkage even for ridgeless estimators. 

Let $\gS=UDU'$ be the population eigenvalue decomposition. Then
\[
(\hat Z_*(z)I+\gS)^{-1}\mu
=
\sum_{i=1}^P
\frac{U_i'\mu}{\gl_i}\,
\frac{1}{1+\hat Z_*(z)\gl_i^{-1}}\,U_i\,,
\]
and we arrive at the following result. 

\begin{theorem}[Projection Theorem]\label{project}
If $\hat Z_*(z)$ separates the spectrum so that $\hat Z_*(z)/\gl_i=O(\eps)$ for $i\le I$ and
$\gl_i/\hat Z_*(z)=O(\eps)$ for $i>I$, then
\[
(\hat Z_*(z)I+\gS)^{-1}\mu
\ =\ P_I(\gS^{-1}\mu)
\ +\ O(\eps)\ ,\ P_I\ =\ \sum_{i=1}^I U_iU_i'\,. 
\]
\end{theorem}
The economic intuition behind Theorem \ref{project} is as follows. Finite samples force investors to downweight directions of factor risk that are
difficult to estimate. This shrinkage intensifies as dimensionality increases. What
matters for performance is therefore not the number of factors, but whether economically
relevant risk premia concentrate in directions that survive shrinkage. Random Features
spread risk premia across many weak directions and are poorly aligned. Learned features,
by contrast, can rotate the factor space so that economically important variation loads
on strong principal components. 
Alignment determines how much of the population Sharpe ratio survives finite-sample
shrinkage. This naturally raises the question of how to estimate $\mu'\gS^{-1}\mu$ when
$P/T$ is non-negligible.

In low dimensions, $\mu'\gS^{-1}\mu$ is usually estimated by the \citep{gibbons1989test} (henceforth, GRS) statistic.
In high dimensions, however, this estimator is biased. \cite{chernov2025test} show that
the bias can be corrected using the same effective penalty $\hat Z_*(z)$ that governs
limits to learning:
\begin{equation}\label{debiased}
\hat W^{debiased}(z)
=
\frac{\bar F'(zI+\hat\gS)^{-1}\bar F-(z^{-1}\hat Z_*(z)-1)}{z^{-1}\hat Z_*(z)}
\ \approx\
\mu'(\hat Z_*(z)I+\gS)^{-1}\mu\,.
\end{equation}
The debiased statistic $\hat W^{debiased}(z)$ therefore provides a natural finite-sample 
measure of economically relevant alignment. In the empirical analysis below, we apply 
this statistic to both random features and PTK-based learned features.

\section{Empirics}\label{sec:empirics}

This section evaluates the empirical implications of the Portfolio Tangent Kernel (PTK) framework using U.S.\ equity data. We compare the performance of three classes of stochastic discount factors (SDFs): (i) fully trained deep neural network (DNN) SDFs, (ii) random-feature SDFs following \cite{didisheim2024apt}, and (iii) PTK-based SDFs that optimally price features learned by the DNN. Our analysis focuses on out-of-sample pricing performance, factor spanning, statistical robustness, and economic alignment with long-horizon consumption risk. All results are reported using rolling-window implementations and are evaluated separately across market capitalization groups.

%%%%%%%%%%%%%%%%%%%%%%%%%%%%%%%%%%%%%%%%%%%%%%%%%%%%%%%%%%%%
\subsection{Data and Random Feature Portfolios}
%%%%%%%%%%%%%%%%%%%%%%%%%%%%%%%%%%%%%%%%%%%%%%%%%%%%%%%%%%%%

We use monthly data on 153 stock characteristics for U.S.\ equities over the period 1963--2024 from \citep*[hereafter JKP]{jensen2023there}. This open-source dataset
collects a comprehensive set of firm-level characteristics drawn from the asset pricing
literature.\footnote{JKP characteristics for individual stocks are available at
\url{https://wrds-www.wharton.upenn.edu/pages/get-data/contributed-data-forms/global-factor-data/}.
Detailed documentation and additional factor portfolio data are provided at
\url{jkpfactors.com}.}
The universe consists of NYSE, AMEX, and NASDAQ stocks with CRSP share codes 10--12.

To construct a stock--month panel with a stable set of characteristics, we follow the
filtering procedure of \cite{didisheim2024apt}. First, we restrict attention to a subset
of 132 characteristics with fewer than 30\% missing observations over the full sample.
Second, we exclude nano stocks whose market capitalization falls below the 1st percentile
of NYSE market capitalizations. Third, we drop stock--month observations with missing
values for more than one-third of the characteristics. Finally, each characteristic is
rank-standardized cross-sectionally to the interval $[-0.5,0.5]$, and remaining missing
values are imputed using the cross-sectional median, which equals zero by construction.
After this procedure, the resulting $N_t\times132$ matrix of characteristics at each date
constitutes the conditioning set $X_t$. Let $R_{t+1}=(R_{i,t+1})_{i=1}^{N_t}$ denote the
vector of next-month excess returns.

To isolate the role of statistical complexity while abstracting from liquidity concerns, we conduct the analysis separately across market capitalization groups. Following JKP,
we study mega (largest 20\%), large (80\%--50\%), small (50\%--20\%), and micro (20\%--1\%)
stocks based on NYSE breakpoints each period, as well as the full universe of stocks above
the 1\% NYSE cutoff, which we refer to as ``all.''

As a benchmark, we construct random-feature-based SDFs following
\cite{didisheim2024apt}. Specifically, we define
\[
S_t={\rm ReLU}(X_t\gamma)\ \in\ \R^{N_t\times P},
\]
where $\gamma=(\gamma_{j,k})\in\R^{d\times P}$ with $P=25{,}000$ and
$\gamma_{j,k}\sim N(0,1)$ independently. These nonlinear transformations constitute the
{\it random features} used in the analysis of \cite{didisheim2024apt}. We then form factor returns
\begin{equation}\label{non-linearDKKM}
F^{\mathrm{ReLU}}_{t+1}
\;=\;
N_t^{-1/2}\,R_{t+1}'S_t\,.
\end{equation}

\noindent\textit{Normalization.} Throughout the paper, all factor returns are normalized by $N_t^{-1/2}$, where $N_t$ denotes the cross-sectional sample size at time $t$. This scaling ensures that factor second moments remain stable across time and are comparable across periods with different cross-sectional dimensions. For the DNN and PTK constructions, this normalization is implemented directly in the factor definitions, ensuring scale comparability with the random feature benchmark.

We use these normalized factor returns to estimate a candidate stochastic discount factor as the return on a Markowitz portfolio computed over a rolling window of $T$ months with ridge penalty $z$:
\begin{equation}\label{R-SDFeLu}
R^{\mathrm{ReLU}}_{t+1}(z,T)\ =\ \theta_t(z;T)'F^{\mathrm{ReLU}}_{t+1},
\end{equation}
where $\theta_t(z;T)$ is obtained from \eqref{msrr-th}.

Because the impact of the ridge penalty depends on its magnitude relative to the factor spectrum, we scale $z$ proportionally to the average eigenvalue of the factor covariance
matrix,
\begin{equation}\label{z=zeff}
z\ =\ z_{eff} P^{-1}\tr(F'F/T)\,.
\end{equation}
We consider rolling estimation windows $T=12,60,120,240,360$ months and set
$z_{eff}=10^{-5}$. The resulting ReLU-based SDF returns are combined into an ensemble
benchmark using weights inversely proportional to their realized volatility over the
preceding 12 months. Importantly, these volatility weights are computed using only
information available at time $t$.

Formally, the ensemble random feature SDF is given by
\begin{equation}\label{the-ReLU-bench}
\bar R^{\mathrm{ReLU}}_{t+1}(z_{eff})
\;=\;
\sum_{T\in\{12,60,120,240,360\}}
\omega_{t,T}\,
R^{\mathrm{ReLU}}_{t+1}(z_{eff} P^{-1}\tr(F'F/T),T),
\end{equation}
where $\omega_{t,T}\propto \widehat{\sigma}^{-1}_{t,T}$ and
$\widehat{\sigma}_{t,T}$ denotes the realized volatility of
$R^{\mathrm{ReLU}}_{\tau+1}(z_{eff},T),\ \tau\le t-1,$ computed over the trailing
12 months.

\subsection{Training the DNN and the PTK-SDF}

We consider a multilayer perceptron (MLP) architecture with an input layer of dimension
132, corresponding to the number of characteristics, followed by $D$ hidden layers of
width $w$. We consider depths $D=1,2,4$ and widths $w=16,32,64,128,256,512$. The parameter
$w$ controls network width, while $D$ controls depth. We denote the resulting function
family by $f(X;\theta;D;w)$.

The dimension of the parameter vector $\theta$ is given by
\[
P \ =\ \underbrace{(d_{in}+1)\times w}_{\text{Input layer}}
\ +\ \underbrace{(D-1)\times (w+1)\times w}_{\text{Hidden layers}}
\ +\ \underbrace{(w+1)}_{\text{Output layer}},
\]
see Definition~\ref{def:mlp}. The deepest and widest network in our analysis contains
$856{,}577$ parameters. Our DNN--PTK duality implies that $P$ is the relevant measure of
model complexity: it equals both the number of factors in the associated PTK and the
interpolation capacity of the network. In particular, by Theorem~\ref{thm:ptk}, the
DNN-SDF can interpolate $T$ observations—achieving in-sample returns equal to one for
$t=1,\ldots,T$—if and only if $P\ge T$, provided the PTK matrix is non-degenerate.

We train the DNN using gradient descent as described in
Appendix~\ref{sec:experiments}, employing a rolling estimation window of 60 months and a
learning rate $\eta=2^{-16}$. At each time $t$, the network is trained using only data
available up to time $t$. To obtain a common set of learned features, we train the network
on the full stock universe (``all'') and obtain a time-$t$ parameter vector
$\theta_t^*$. We then evaluate the resulting SDF separately on each size group (mega,
large, small, and micro).

The resulting DNN-based SDF is
\begin{equation}\label{def:rdnn}
R^{\mathrm{DNN}}_{t+1}(D;w)
\;=\;
N_t^{-1/2}\,R'_{t+1} f(X_t;\theta^*_t;D;w)\,.
\end{equation}

We next construct SDFs based on the Portfolio Tangent Kernel. Given a trained network
$f(X;\theta_t^*;D;w)$, we form factor returns using the learned gradient features as in
\eqref{factor-ret},
\begin{equation}
F_{\tau+1}(\theta_t^*;D;w)
\ =\ N_\tau^{-1/2}\,
R_{\tau+1}'\underbrace{\nabla_\theta f(X_\tau;\theta_t^*;D;w)}_{time\text{-}t\ learned\ features}
\ \in\ \R^{P},\qquad P=\dim(\theta)\,.
\end{equation}
Thus, $F_{\tau+1}(\theta_t^*;D;w)$ represents returns on the time-$t$ vintage of
characteristics-managed portfolios constructed from features learned at time $t$.

The corresponding time-$t$ Portfolio Tangent Kernel (PTK) is defined as
\begin{equation}
\bK\big((R_{\tau_1+1},X_{\tau_1});(R_{\tau_2+1},X_{\tau_2});\theta_t^*;D;w\big)
\ =\ F_{\tau_1+1}(\theta_t^*;D;w)' F_{\tau_2+1}(\theta_t^*;D;w)\,.
\end{equation}

Fixing a rolling window length $T$, we estimate the ridge-penalized Markowitz portfolio
\eqref{msrr-th} using factor returns
$F_{\tau+1}(\theta_t^*;D;w)$ for $\tau=t-T,\ldots,t-1$. Denoting the resulting portfolio
weights by $\Theta_t(\theta_t^*;D;w;T;z)$, the associated PTK-based SDF is
\begin{equation}\label{learned-feat-1}
R^{\mathrm{PTK}}_{t+1}(D;w;T;z)
\ =\
\Theta_t(\theta_t^*;D;w;T;z)'F_{t+1}(\theta_t^*;D;w)\,.
\end{equation}

As with the random feature benchmark, we construct an ensemble PTK-SDF by averaging across
rolling windows,
\begin{equation}\label{the-ptk-bench}
\bar R^{\mathrm{PTK}}_{t+1}(D;w;z_{eff})
\;=\;
\sum_{T\in\{12,60,120,240,360\}}
\omega_{t,T}^{\mathrm{PTK}}\,
R^{\mathrm{PTK}}_{t+1}(D;w;T;z_{eff} P^{-1}\tr(F'F/T)),
\end{equation}
where $\omega_{t,T}^{\mathrm{PTK}}\propto \widehat{\sigma}^{-1}_{t,T}$ and
$\widehat{\sigma}_{t,T}$ denotes the realized volatility of
$R^{\mathrm{PTK}}_{\tau+1}(D;w;T;z),\ \tau\le t-1,$ computed over the trailing 12 months.

This aggregation exploits potential non-stationarity by allowing the pricing rule to adapt across time scales. Importantly, this flexibility is specific to the PTK
construction and entails negligible computational cost, as estimating Markowitz
portfolios for multiple $T$ is numerically trivial.

\subsection{The Performance of $R^{\rm ReLU}, R^{\rm DNN},$ and $R^{\mathrm{PTK}}$}

The random feature portfolios $R^{\mathrm{ReLU}}_{t+1}$ defined in
\eqref{R-SDFeLu} correspond to a shallow (single hidden layer) neural network
with fixed, randomly generated features. Despite the absence of feature
learning, \cite{didisheim2024apt} document strong empirical performance of this
nonlinear SDF, driven by sheer statistical complexity: the model effectively
prices returns using $P=25{,}000$ nonlinear factors. This benchmark, therefore,
provides a stringent baseline against which to evaluate fully trained neural
networks.

Our first objective is to assess whether a trained DNN is able to {\it learn
economically relevant features}—that is, priced nonlinear characteristics that
are not already captured by the random feature SDF.

Figure~\ref{fig:mlp_sharpe_width_depth} reports the Sharpe ratios of
$R^{\rm DNN}_{t+1}(D;w)$ from \eqref{def:rdnn} as a function of network depth and
width, alongside the ensembled ReLU benchmark
$\bar R^{\mathrm{ReLU}}$ from \eqref{the-ReLU-bench}. The patterns are similar across size groups; we therefore focus on the full stock universe (``all'').

Several findings stand out. First, depth is essential. Shallow networks
($D=1$) perform poorly and substantially underperform the random feature
benchmark. Second, deeper networks are capable of extracting more structure:
For example, the depth-2, width-256 network attains a Sharpe ratio of 3.8,
slightly exceeding the ReLU benchmark (3.7). This architecture is highly
complex, with $P=100{,}097$ trained parameters. Nevertheless, even such large
networks generally fail to outperform the random feature SDF across all size
groups, despite having four times as many parameters.

Taken at face value, these results might suggest that fully trained DNNs add
little value for asset pricing. This conclusion, however, is misleading. As
emphasized by \cite{kelly2025understanding}, raw Sharpe ratios alone are not an
appropriate metric for evaluating machine learning models. What matters is
whether the model uncovers {\it new priced variation}—that is, whether it
discovers risk premia that are unspanned by existing models.

To separate {\it what} the DNN learns from {\it how} it prices, we test the
spanning hypothesis by estimating the univariate regressions on the \eqref{the-ReLU-bench} benchmark with $z_{eff}=10^{-5}:$
\begin{equation}\label{main-regr0}
R_{t+1}^{\rm DNN}(D;w)
\ =\
\ga
\ +\
\beta\,\bar R^{\mathrm{ReLU}}_{t+1}(10^{-5})
\ +\
\eps_{t+1}\,.
\end{equation}
Figure~\ref{fig:mlp_alpha_tstat_comparison} reports the $t$-statistics of the
intercepts.

The results strongly support the feature-learning hypothesis. For moderately complex deep architectures (e.g., $D=2$ and $D=4$ with $w=256$), the estimated alpha is positive and highly statistically significant. In contrast,
very small networks fail to uncover new priced characteristics, while extremely
large networks (depth 4 with a width above 256, corresponding to
$P>200{,}000$) sometimes fail to do so reliably.\footnote{Training such large
networks may require alternative optimization schemes or architecture-adjusted
hyperparameters; see \cite{yang2022tensor}. We leave this issue for future work.}

Economically, these findings imply that moderately complex DNNs
($P\approx100{,}000$) are effective at learning {\it what to price}: they
discover a rich set of priced nonlinear characteristics that are not spanned by
random features. However, they are less effective at learning {\it how to price}
these characteristics, resulting in suboptimal Sharpe ratios despite strong
alpha.

Motivated by these results, we focus on the architecture
$R^{\rm DNN}(D=2,w=256)$ in the remainder of the analysis. This specification
exhibits the strongest and most robust feature-learning performance across all
size groups, as illustrated in
Figure~\ref{fig:mlp_alpha_tstat_comparison}.

We now turn to our central theoretical prediction: the Portfolio Tangent Kernel
(PTK) should optimally price the features learned by the DNN. If the learned
characteristics subsume those captured by the ReLU benchmark, and if the PTK
implements an efficient pricing rule, then
$\bar R^{\rm PTK}(D=2,w=256)$ should outperform both $R^{\rm DNN}$ and
$\bar R^{\mathrm{ReLU}}$.

Figure~\ref{fig:ptk_sharpe_lambda} confirms this prediction. The PTK-SDF delivers
substantially higher Sharpe ratios than both benchmarks across all size groups.
The gains are particularly pronounced for mega and large stocks: PTK achieves Sharpe
ratios of 1.6 and 1.9, compared to 1.3 and 1.6 for the ReLU benchmark.

A striking implication is that the {\it ridgeless} PTK attains the highest
Sharpe ratios across all size groups. 
These findings underscore the role of statistical complexity: when low-variance
principal components contain an economically meaningful signal—as suggested by the
strong alpha evidence in Figure~\ref{fig:mlp_alpha_tstat_comparison}—explicit shrinkage becomes unnecessary.

Figure~\ref{fig:alpha_tstat_comparison} further shows that the outperformance of
PTK is statistically significant. The PTK-SDF delivers large and significant
alpha relative to both $R^{\rm DNN}$ and $\bar R^{\mathrm{ReLU}}$, indicating that it
fully subsumes their pricing information.

We therefore conclude that feature learning and pricing should be viewed as
distinct tasks. DNNs excel at discovering rich nonlinear characteristics, while
the PTK provides a statistically efficient mechanism for pricing them. The PTK
representation renders the learned SDF amenable to the Random Matrix Theory
analysis developed in Section~\ref{sec:alignment}, to which we now turn.

\begin{figure}[htbp]
\centering
\includegraphics[width=0.8\textwidth]{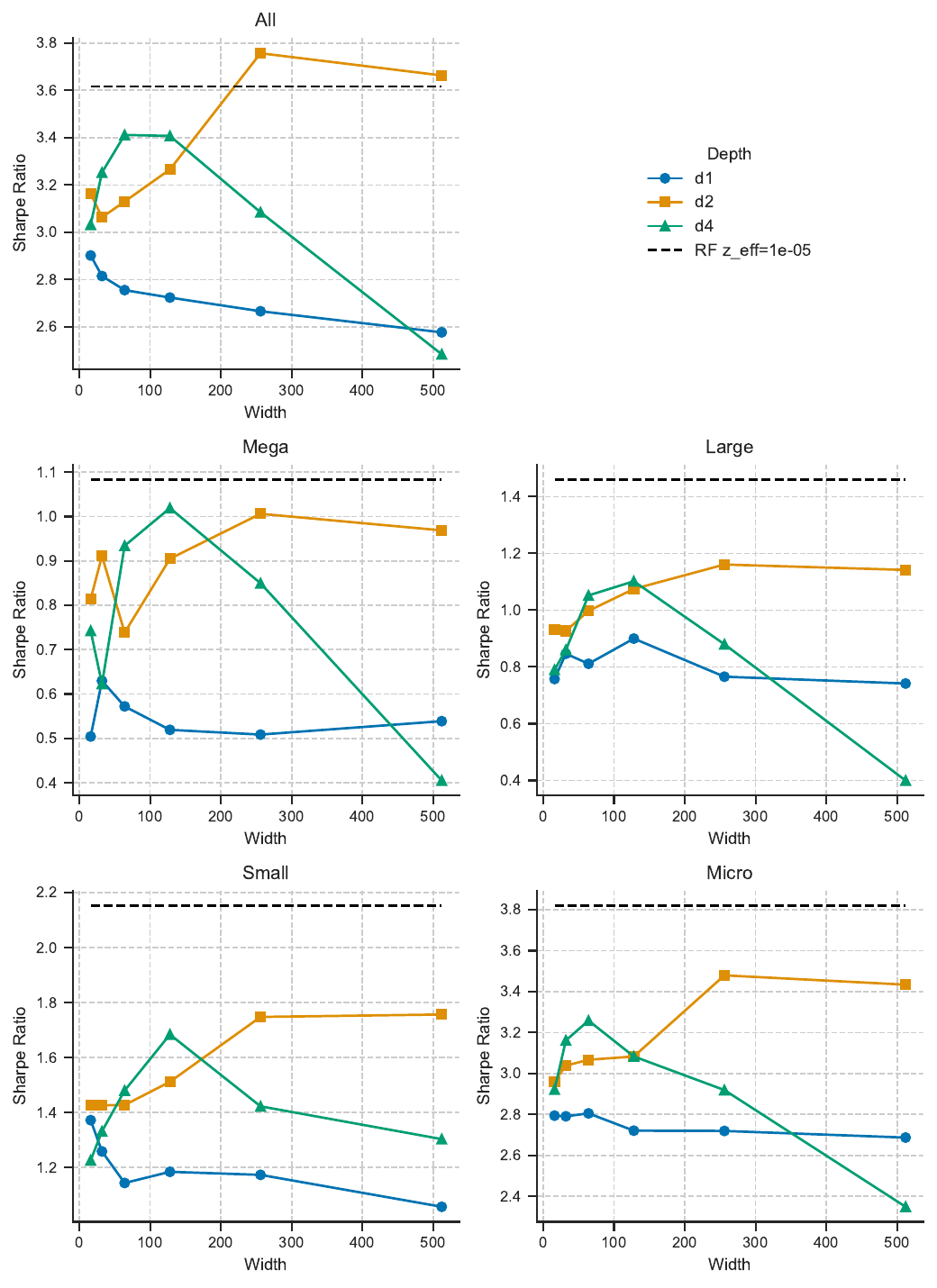}
\caption{Sharpe Ratio of $R^{\rm DNN}$ from \eqref{def:rdnn} as a Function of Model Complexity (depth $D$ and width $w$). The lines $d1,d2,d4$ correspond to $D=1,2,$ and $4$, respectively, with widths $w \in \{2^4, \dots, 2^9\}$. The dashed line is the Sharpe ratio of the $\bar R^{\mathrm{ReLU}}$ benchmark from \eqref{the-ReLU-bench} with  $z_{eff}=10^{-5}$.} 
\label{fig:mlp_sharpe_width_depth}
\end{figure}

\begin{figure}[htbp]
\centering
\includegraphics[width=0.8\textwidth]{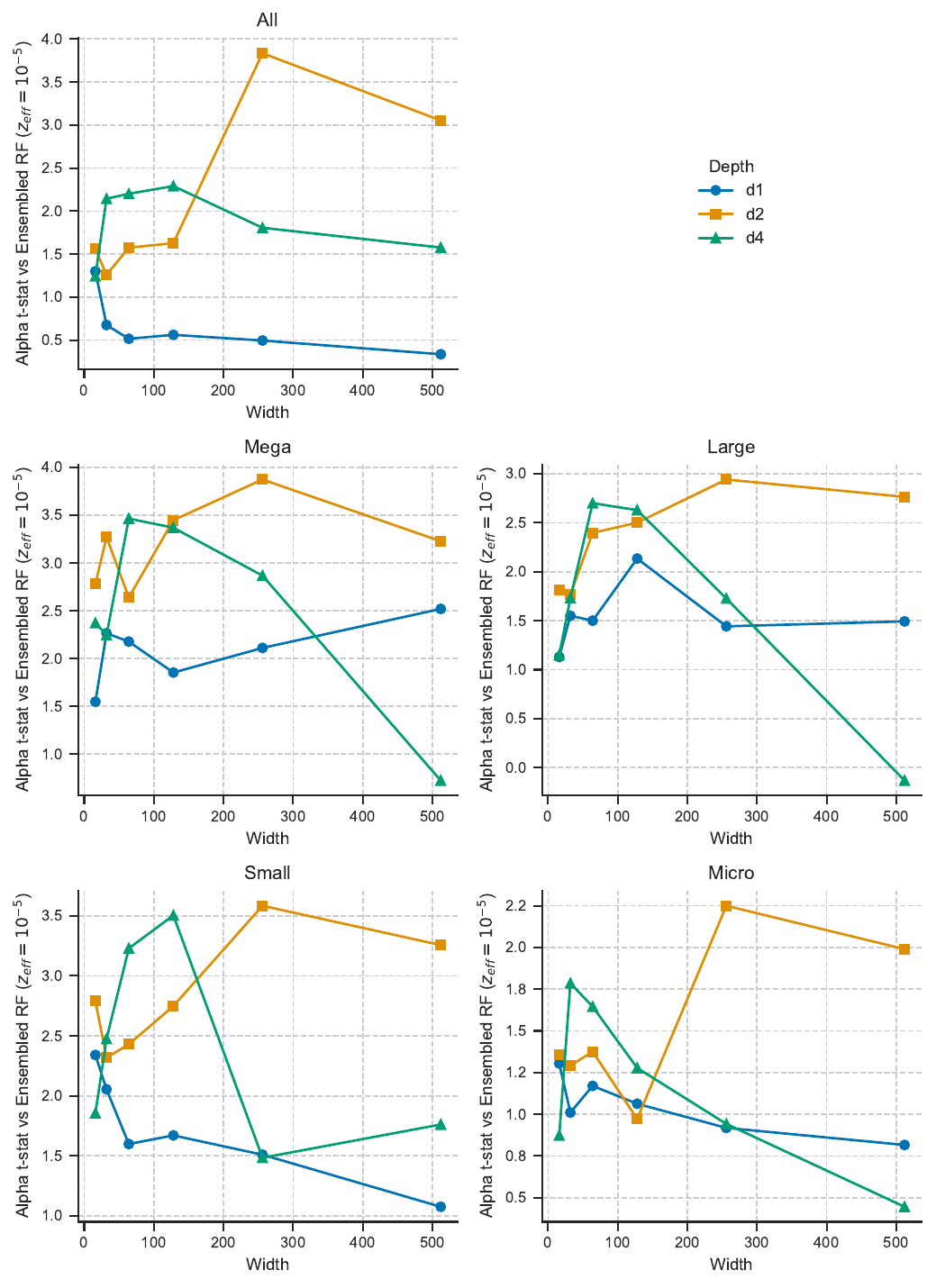}
\caption{Alpha t-statistics from the regressions \eqref{def:rdnn} as a Function of Model Complexity (depth $D$ and width $w$). The lines $d1,d2,d4$ correspond to $D=1,2,$ and $4$, respectively, with widths $w \in \{2^4, \dots, 2^9\}$. }
\label{fig:mlp_alpha_tstat_comparison}
\end{figure}

\begin{figure}[htbp]
\centering
\includegraphics[width=0.8\textwidth]{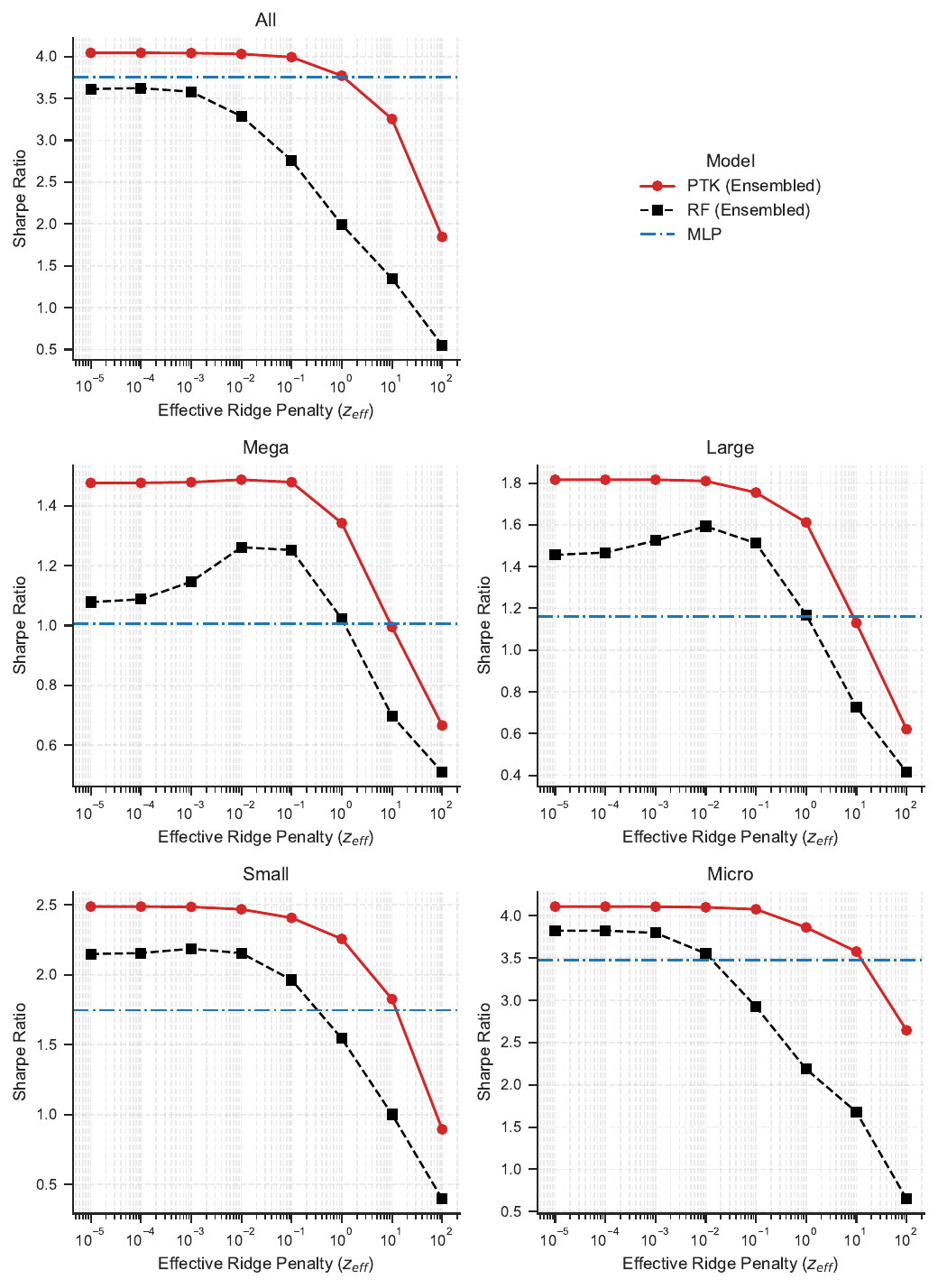}
\caption{Sharpe Ratios of $\bar R^{\rm PTK}(D=2,w=256;z_{eff})$ from \eqref{the-ptk-bench} and $\bar R^{\rm ReLU}(z_{eff})$ from \eqref{the-ReLU-bench} (both are ensembled across $T$), as a function of the $z_{eff}.$ The blue dash-dotted line represents the Sharpe Ratio of $R^{\rm DNN}(D=2,w=256)$ benchmark from \eqref{def:rdnn}.}
\label{fig:ptk_sharpe_lambda}
\end{figure}

\begin{figure}[htbp]
\centering
\includegraphics[width=0.8\textwidth]{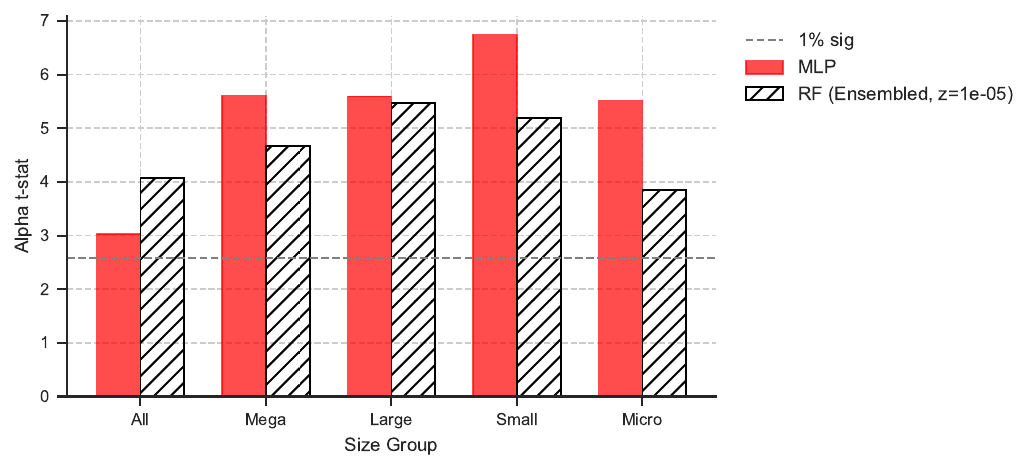}
\caption{The t-statistics of alpha from univariate regressions of $\bar R^{\rm PTK}(D=2,w=256)$ with $z_{eff}=10^{-5}$ (see \eqref{the-ptk-bench}) on $R^{\rm DNN}(D=2,w=256)$ from \eqref{def:rdnn} (red bars) and $\bar R^{\textrm{ReLU}}$ from \eqref{the-ReLU-bench} (grey striped bars). The dashed horizontal line indicates the level of 1\% statistical significance.}
\label{fig:alpha_tstat_comparison}
\end{figure}
This out-performance is also statistically significant, as is shown by Figure \ref{fig:alpha_tstat_comparison}. We therefore conclude that learned features (PTK) entirely subsume the performance of the DNN.  In turn, the PTK representation of the candidate SDF makes it amenable to the RMT analysis from Section \ref{sec:alignment}.

\subsection{The Dynamics of Spectral Complexity}

By Theorem~\ref{thm:master}, a central determinant of LFM performance is the
effective shrinkage $\hat Z_*(z)$ defined in \eqref{dfn:zstar}. In the PTK
setting, we adopt a scaled version of \eqref{z=zeff} and set
\begin{equation}
z\ =\ z_{eff}\,T^{-1}\tr(FF'/T)\,,
\end{equation}
which is proportional to \eqref{z=zeff}. Since $\hat Z_*(z)$ scales linearly with the kernel matrix $K=FF'$, we normalize the effective ridge penalty and define
\begin{equation}\label{z+1}
\frac{\hat Z_*(z)}{T^{-2}\tr(K)}
\ =\
\frac{1}{T^{-1}\tr\!\big((z_{eff}T^{-2}\tr(K) I+K/T)^{-1}\big)\,\cdot\,T^{-2}\tr(K)}\,.
\end{equation}
By Jensen’s inequality, $T^{-1}\tr(A^{-1})\ \ge\ (T^{-1}\tr(A))^{-1}$ for any
symmetric positive definite $T\times T$ matrix $A$, which implies
\begin{equation}\label{z+1ineq}
\frac{\hat Z_*(z_{eff}\,T^{-2}\tr(K))}{T^{-2}\tr(K)}
\ \le\ 1+z_{eff}\,.
\end{equation}
The normalized quantity in \eqref{z+1} measures the fraction of total factor
variance effectively shaded by implicit regularization in an
over-parameterized LFM. The upper bound in \eqref{z+1ineq} is attained when
$K\approx I$, corresponding to maximal spectral complexity, where all principal
components contribute equally. We therefore refer to \eqref{z+1} as the
{\it spectral complexity} of the kernel matrix.

Figure~\ref{fig:ptk_implicit_grs} reports the time-series dynamics of spectral
complexity (left panel) and the debiased GRS statistic \eqref{debiased} (right
panel) for both random-feature-based factors \eqref{non-linearDKKM} and PTK-based
factors. Several striking empirical regularities emerge.

First, PTK factors exhibit spectral complexity roughly an order of magnitude
larger than that of random features. Rather than concentrating risk in a small
number of dominant directions, the learned representation spreads economically
relevant variation across a large number of principal components. This finding
underscores that feature learning does not simplify the factor space; instead,
it reorganizes it.

Second, PTK spectral complexity rises monotonically after 2000 and increases by
approximately a factor of six over the sample period. In contrast, the spectral
complexity of random feature factors remains essentially flat. The increase is
therefore specific to learned, priced features and reflects a gradual rise in
the dimensionality of economically relevant risk.

Third, the debiased GRS statistic \eqref{debiased} rises steadily until the early
2000s and declines after 2010. Importantly, this decline is not accompanied by a
deterioration in realized SDF performance (Figure~\ref{fig:ptk_cumulative_returns} in the Appendix).
Combined with the sharp increase in $\hat Z_*(z)$, this pattern suggests that
the population Sharpe ratio $\mu'\gS^{-1}\mu$ has continued to grow, while
finite-sample limits to learning have intensified due to rising spectral
complexity.

To further investigate alignment in the sense of Theorem~\ref{project}, we
conduct a principal-component truncation exercise. For each $K$, we construct
truncated versions of $R^{\mathrm{PTK}}_{t+1}(D_*;w_*;T,z)$ and
$R^{\mathrm{ReLU}}_{t+1}(T,z)$ by retaining only the top $K$ eigenvectors of the
sample covariance matrix
\begin{equation}\label{eigen-dec}
\hat\gS_t
\ =\
T^{-1}\sum_{\tau=t-T}^{t-1}F_{\tau+1}F_{\tau+1}'
\ =\
\hat U_t \hat D_t \hat U_t'\,,
\end{equation}
and forming portfolios based on 
\begin{equation}\label{ftk}
\hat F_t(K)\ =\ \hat U_t[:,:K]'\hat F_t.
\end{equation}
Figure~\ref{fig:topk_sharpe_comparison} plots Sharpe ratios as a function of $K$
for learned gradient features and random features. Consistent with the alignment
mechanism of Section~\ref{sec:alignment}, learned features concentrate
economically relevant signal in far fewer directions. For the full stock
universe, the top $15$ (out of $60$) PTK principal components already achieve a
Sharpe ratio of $3.8$, close to the full-feature maximum of $3.9$. By contrast,
random features require substantially larger $K$ to extract the signal and reach
only $2.3$ at the same truncation level.

Taken together, these results highlight a central tension. Feature learning
increases spectral complexity and exacerbates limits to learning, yet
simultaneously improves alignment by concentrating the economic signal in
statistically robust directions. The empirical success of PTK-based pricing
stems precisely from this trade-off.

\begin{figure}[htbp]
\centering
\includegraphics[width=0.75\textwidth]{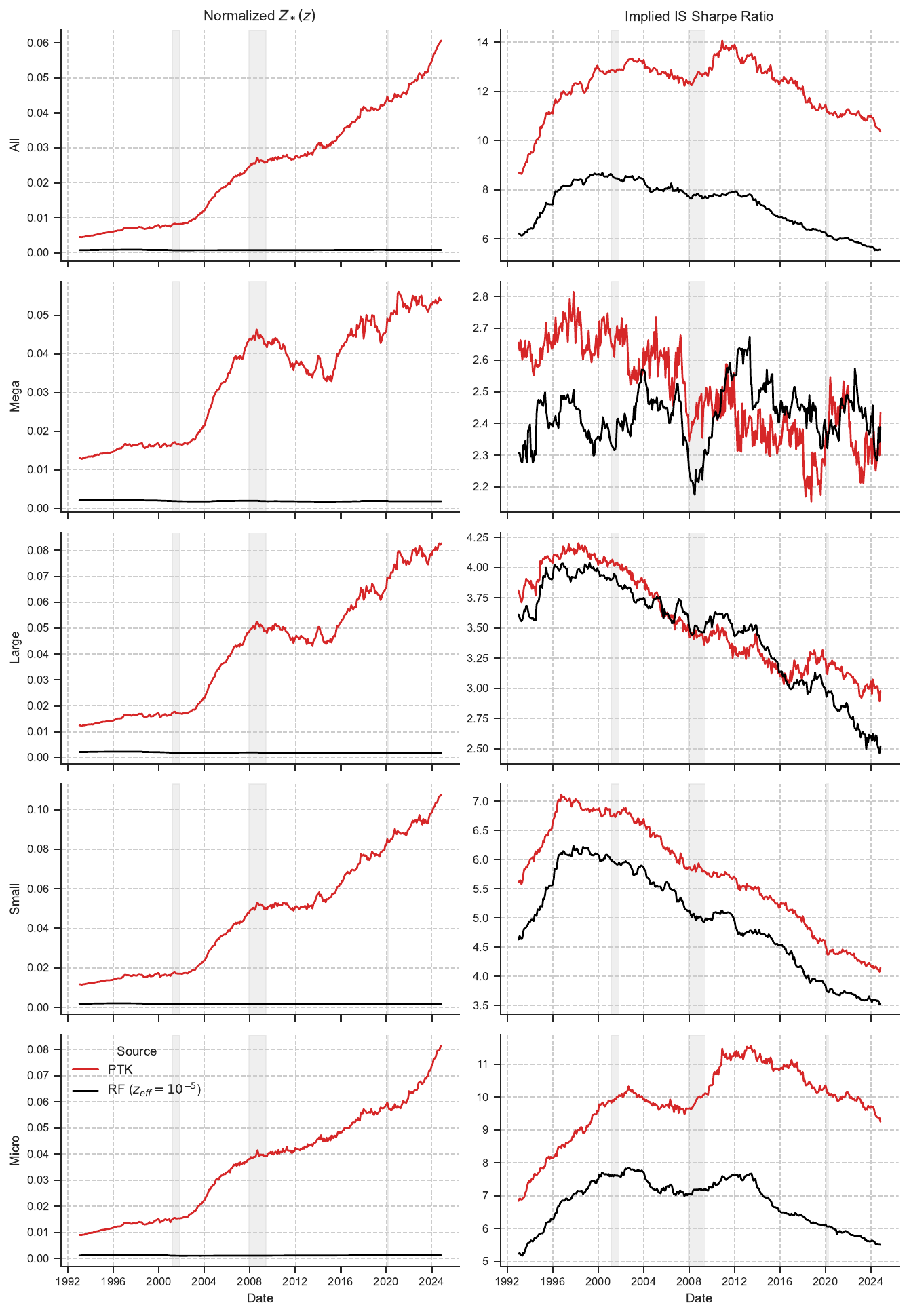}
\caption{The dynamics of spectral complexity \eqref{z+1} and the debiased GRS statistic \eqref{debiased} estimated using $T=360$ for $R^{\rm PTK}$ and $R^{\textrm{ReLU}}$; we use $D=2,\ w=256.$}
\label{fig:ptk_implicit_grs}
\end{figure}

\begin{figure}[htbp]
\centering
\includegraphics[width=0.8\textwidth]{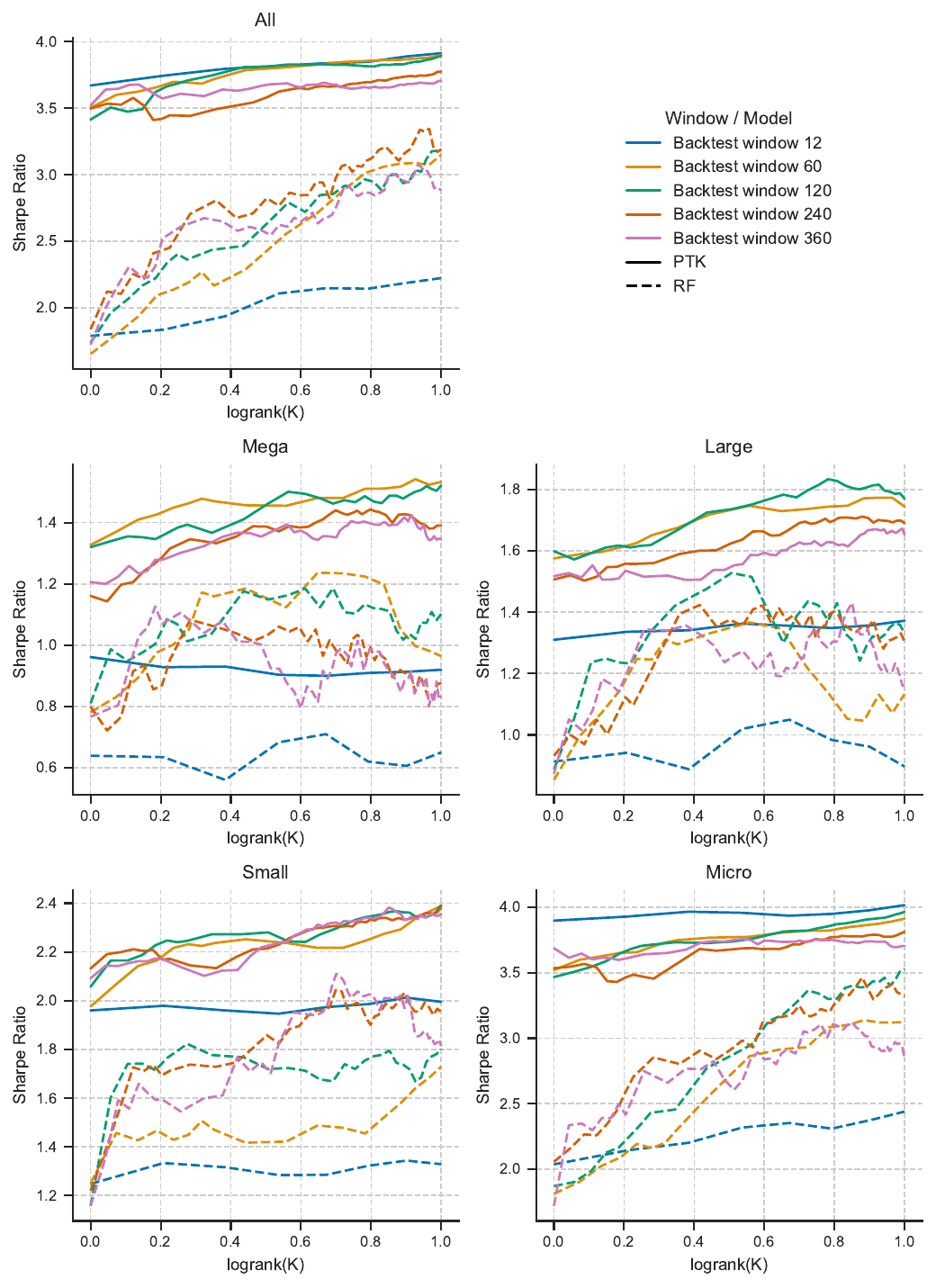}
\caption{Factor Alignment. This figure plots the Sharpe Ratio against the number of principal components ($K$) for PTK Factors and random feature factors, penalized with $z_{eff}=10^{-5}$, for different backtest window sizes $T=12, 60, 120, 240, 360$.}
\label{fig:topk_sharpe_comparison}
\end{figure}

\subsection{PTK Alignment with Consumption Risk}

Modern asset pricing theories emphasize long-run and future consumption risk as key drivers of equity premia \citep{bansal2004risks,campbell2018intertemporal}. In these models, the stochastic discount factor (SDF) loads on future consumption realizations rather than contemporaneous consumption growth.

\cite{ParkerJulliard2005} show that covariances with future realized consumption are priced in the cross-section of stock returns, while \cite{MalamudWangZhang2025} provide a micro-founded version based on financial frictions. These frameworks offer economically disciplined benchmarks for evaluating whether machine-learning-based SDFs capture meaningful macroeconomic risk.

Our alignment framework implies that economically relevant pricing information should concentrate in the leading principal components of the factor space (Theorem~\ref{project}). We therefore examine whether the dominant components of PTK-based factors exhibit stronger alignment with future consumption risk than those of random features or fully trained DNN-SDFs.

We project $\bar R^{\rm PTK}$ and $\bar R^{\rm ReLu}$ onto their top two principal components \eqref{ftk} and compute correlations and predictive regressions with future consumption growth, following \cite{ParkerJulliard2005} and \cite{MalamudWangZhang2025}.\footnote{Results without principal-component projection are weaker, indicating that lower-variance PTK components primarily capture variation orthogonal to consumption risk. See Figure \ref{fig:links_to_ccapm} in the Appendix.}

Figure~\ref{fig:links_to_ccapm_K2} reports the results. Across all size groups, the PTK-SDF exhibits substantially stronger alignment with both consumption-based benchmarks than either the DNN-SDF or the random-feature benchmark. PTK delivers larger absolute correlations, higher regression $t$-statistics, and materially greater explanatory power. For the aggregate stock universe, the projected PTK-SDF is approximately $-30\%$ correlated with future consumption growth under \cite{ParkerJulliard2005} and nearly $-50\%$ under \cite{MalamudWangZhang2025}, compared to correlations around $-10\%$ and $-20\%$, respectively, for the alternatives.

These findings highlight the economic role of PTK. Gradient-based feature learning rotates the high-dimensional characteristic space so that consumption-related risk loads on statistically strong directions that survive shrinkage. Random features fail to achieve this alignment, while unconstrained DNN-SDFs do not price it efficiently.

Overall, PTK isolates the economically relevant component of feature learning and translates it into a disciplined pricing rule, yielding tradable SDFs that are tightly linked to long-horizon consumption risk.

\begin{figure}[htbp]
\centering
\includegraphics[width=0.8\textwidth]{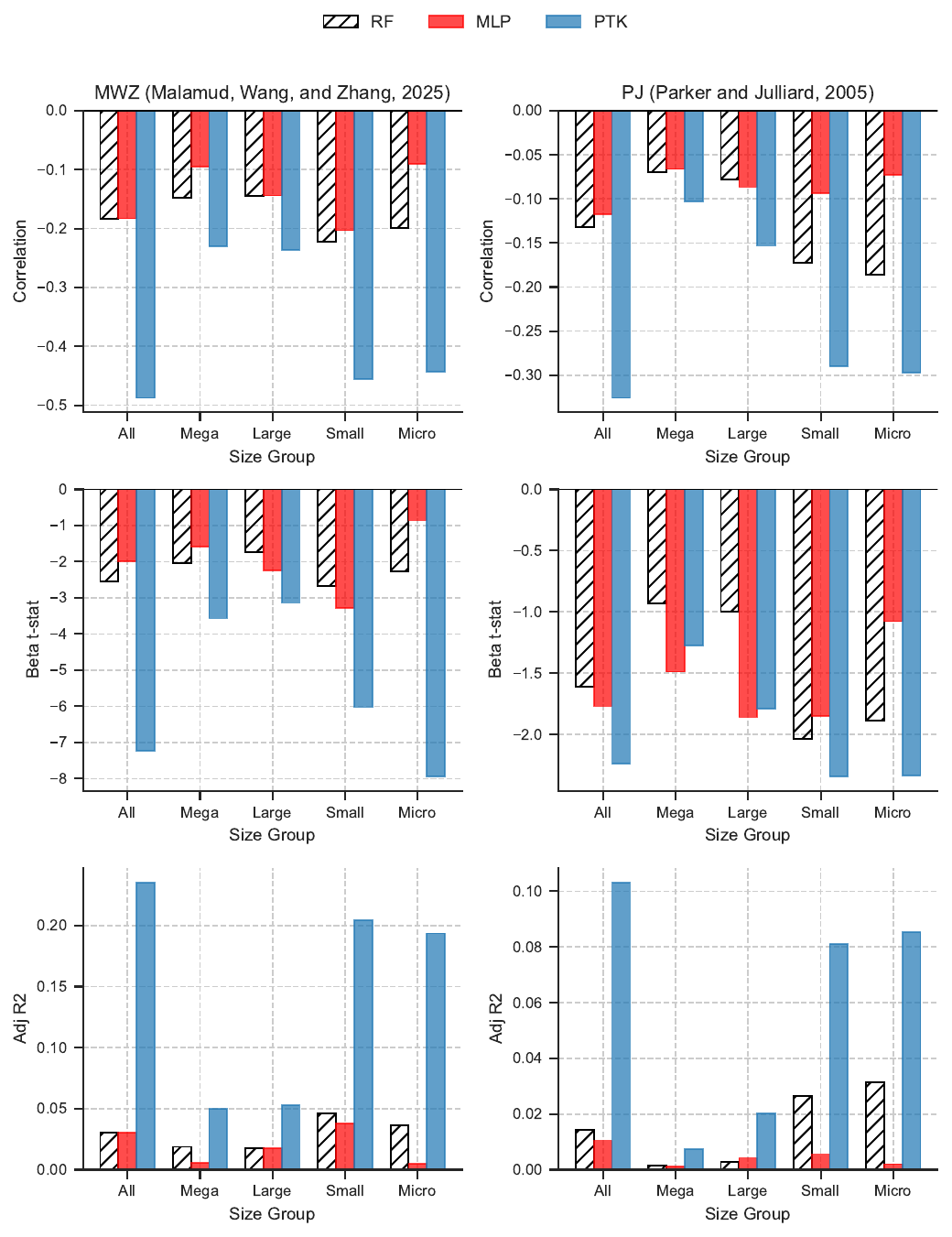}
\caption{(Top 2 PCs for ReLU and PTK based SDFs) Machine Learning SDFs and Consumption Risk.
This figure reports the relation between tradeable machine-learning stochastic discount factors (SDFs) and consumption-based SDFs from \cite{ParkerJulliard2005} (PJ) and \cite{MalamudWangZhang2025} (MWZ). The top row shows time-series correlations between portfolio returns from the fully trained DNN SDF $R^{\mathrm{DNN}}$ \eqref{def:rdnn} (red), the random-feature benchmark $\bar R^{\mathrm{ReLU}}$ \eqref{the-ReLU-bench} (grey, striped), and the PTK-based SDF $\bar R^{\mathrm{PTK}}$ \eqref{the-ptk-bench} (blue), and the MWZ future-consumption SDF (left panel) and the PJ long-run consumption SDF (right panel). 
The middle row reports $t$-statistics from univariate time-series regressions of each consumption-based SDF on the corresponding tradeable SDF return. The bottom row reports the associated adjusted $R^2$. All statistics are shown separately by size group. Consumption-based SDFs are constructed using future realized consumption growth aggregated over horizons of 84 months (MWZ) and 36 months (PJ), respectively, which truncates the sample in June 2018. Ridge regularization with $z_{eff}=10^{-5}$ is applied to the ReLU and PTK portfolios.}
\label{fig:links_to_ccapm_K2}
\end{figure}

\section{Conclusion}

This paper studies how modern machine learning models price assets in
high-dimensional environments and why their performance is constrained by
fundamental limits to learning. We develop a framework that links deep neural
networks, kernel methods, and factor pricing through the Portfolio Tangent Kernel
(PTK), which provides an explicit representation of the pricing rule induced by
gradient-based training.

Our main conceptual contribution is to separate feature learning from pricing.
Fully trained neural networks are effective at discovering rich, nonlinear
characteristics that are unspanned by standard factor constructions. However,
when used directly as stochastic discount factors, these models often fail to
aggregate the learned features efficiently. The PTK makes this distinction
explicit: network architecture and training determine which features are learned,
while the PTK governs how these features are priced through a statistically
disciplined factor model.

Using random matrix theory, we show that finite-sample estimation induces
endogenous shrinkage even in ridgeless, interpolating regimes. This shrinkage
imposes an effective low-dimensional structure on the pricing rule, and its
severity depends critically on the alignment between risk premia and the
principal components of the feature space. We formalize this mechanism through a
projection theorem that clarifies how limits to learning translate into economic
restrictions on achievable Sharpe ratios.

Empirically, we find that learned gradient features exhibit substantially stronger
alignment than random features, allowing the PTK-based SDF to outperform both raw
DNN-SDFs and strong random feature benchmarks across asset size groups. 
The results show that the gains from deep learning arise not from abandoning
structure, but from combining flexible representation learning with explicit and
statistically efficient pricing rules.

Overall, our findings suggest that the role of machine learning in asset pricing
is best understood as expanding the space of economically relevant
characteristics, while leaving pricing to be governed by well-understood
principles of risk, regularization, and finite-sample discipline.

\clearpage
\bibliographystyle{ecta}
\bibliography{biblio}

\appendix

\section{MLPs and the NTK in the Infinite Width Limit}

\begin{definition}[Multi-Layer Perceptron (MLP)]\label{def:mlp} Fix a neural network architecture given by the widths of layers $(n_1,\cdots,n_L).$ Let $\theta=(W^{(0)}, b^{(0)}, \cdots, W^{(L-1)}, b^{(L-1)})$ be a collection of weights and biases. Here, $W^{(l)}\in \R^{n_l\times n_{l+1}},\ b^{(l+1)}\in \R^{n_{l}}$, and the total dimension of $\theta$ is $P=\sum_{l=0}^{L-1} (n_l+1)n_{l+1}$. Let also $\phi:\R\to\R$ be a Lipschitz, twice differentiable nonlinearity with a bounded second derivative. 
The MLP neural network $f(x;\theta)$ is defined as 
\begin{equation}
\begin{aligned}
x &= \text{input}\ \in\ \R^d \\
y^{(l)}(x) &= \begin{cases}
x & \text{if } l = 0, \\
\phi(z^{(l)}(x)) & \text{if } l > 0.
\end{cases} \\
z^{(l+1)}_i(x) &= \frac1{n_l^{1/2}}\sum_j W^{(l)}_{ij} y^{(l)}_j(x) + b^{(l)}_i,\\
\end{aligned}
\end{equation}
where $y^{(l)}(x)$, $z^{(l)}(x) \in \mathbb{R}^{n_l}.$
The output of the network is 
\begin{equation}\label{nngp-1}
f(x;\theta)\ =\ z^{(L)}(x)\ =\ \frac1{n_{L-1}^{1/2}}\sum_{j=1}^{n_{L-1}}W_j^{(L-1)}y_j^{(L-1)}(x)\,. 
\end{equation}
At initialization, all weights $W^{(\ell)}_{i,j}$ and biases $b^{(l)}$ in each hidden layer are sampled independently from $N(0,(\gs_W^{(l)})^2)$ and $N(0,(\gs_b^{(l)})^2)$, respectively, for some layer-specific parameters $\gs_W^{(l)}, \gs_b^{(l)}$.
\end{definition}

As a benchmark example, consider a single-hidden-layer NN 
\begin{equation}\label{simple-nn}
\begin{aligned}
&f(x;\theta)\ =\ z^1(x)\ =\ \frac1{n_1^{1/2}}\sum_{k=1}^{n_1}W_k^{1}\phi\left(\frac1{n_0^{1/2}}\sum_{i=1}^d W_{k,i}^0 x_i+b_k^0\right)\\
&=\ \frac1{n_1^{1/2}}\sum_{k=1}^{n_1}W_k^{1}\phi\left(\frac1{n_0^{1/2}}(W_k^0)'x+b_k^0\right)\\
&=\ \frac1{n_1^{1/2}} (W^1)'\phi\left(\frac1{n_0^{1/2}}(W^0)'x+b^0\right)\,. 
\end{aligned}
\end{equation}
The simple random feature model is therefore equivalent to the shallow (single hidden layer) neural network \eqref{simple-nn}, with $\theta=W^1\in \R^P,\ W_k=W_k^0\in \R^d,\ P=n_1$. However, there is one important caveat: Neural networks are trained end-to-end; that is, every single weight, including all weights of all hidden layers, is optimized. E.g., for the simple 1-hidden layer net \eqref{simple-nn}, the whole $\theta=(W^1,W^0,b^0)\in \R^{n_1(d+2)}$ is optimized upon.

\begin{definition}[The Infinite Width NTK]\label{def:inf-ntk} Let 
\begin{equation}\label{scary-int-1}
\begin{aligned}
&\widehat\Sigma^{(l)}(x,\tilde x) \\
&= (\gs_w^{(l)})^2\int dz \, d\tilde z \frac{d}{dz}\phi(z) \frac{d}{d\tilde z}\phi(\tilde z) \mathcal{N}\left( \begin{pmatrix} z \\ \tilde z \end{pmatrix} ; \mathbf{0}, (\gs_w^{(l-1)})^2\begin{pmatrix} \Sigma^{(l-1)}(x, x) & \Sigma^{(l-1)}(x,\tilde x) \\  \Sigma^{(l-1)}(\tilde x,x) &  \Sigma^{(l-1)}(\tilde x,\tilde x) \end{pmatrix} +(\sigma^{(l-1)}_b)^2\,{\bf 1}_{2\times 2}  \right)\,.
\end{aligned}
\end{equation}
Let also 
\begin{equation}
\Theta^{(1)}(x,\tilde x)\ =\ \Sigma^{(1)}(x,\tilde x)+1,
\end{equation}
and then define recursively for $l>1:$ 
\begin{equation}\label{ntk-closed}
\Theta^{(l)}(x,\tilde x)\ =\ \Theta^{(l-1)}(x,\tilde x)\widehat\Sigma^{(l)}(x,\tilde x)\ +\ \Sigma^{(l)}(x,\tilde x)+1
\end{equation}
\end{definition}

\begin{theorem}\label{ntk:inf} Suppose that $\phi(x)$ is uniformly Lipschitz.\footnote{That is, $|\phi(x)-\phi(y)|\le C |x-y|$ for some $C>0.$} Then, at initialization, 
\begin{equation}
\nabla_{\theta}z^{(l)}_i(x)\nabla_{\theta}z^{(l)}_{i_1}(\tilde x)\ \to\ \gd_{i,i_1}\Theta^{(l+1)}(x,\tilde x)
\end{equation}
where $\Theta$ is defined in Definition \ref{def:inf-ntk}.
\end{theorem}

\begin{proof}[Proof of Theorem \ref{ntk:inf}] The proof is by induction. Recall that 
\begin{equation}
\begin{aligned}
x &= \text{input}\ \in\ \R^d \\
y^{(l)}(x) &= \begin{cases}
x & \text{if } l = 0, \\
\phi(z^{(l-1)}(x)) & \text{if } l > 0.
\end{cases} \\
z^{(l)}_i(x) &= \frac1{n_l^{1/2}}\sum_j W^{(l)}_{ij} y^{(l)}_j(x) + b^{(l)}_i,\\
\end{aligned}
\end{equation}
where $y^{(l)}(x)$, $z^{(l-1)}(x) \in \mathbb{R}^{n^{(l)} \times 1}.$ 
We have 
\begin{equation}
z^{(0)}_i(x)\ =\ \frac1{n_0^{1/2}}\sum_{j=1}^{n_0}W_{i,j}^{(0)}x_j\ +\ b^0_j
\end{equation}
and, hence, 
\begin{equation}
\nabla_{\theta}z^{(0)}_i(x)\ =\ (\frac1{n_0^{1/2}}x,1),
\end{equation}
so that 
\begin{equation}
\nabla_{\theta}z^{(0)}_i(x)\nabla_{\theta}z^{(0)}_{i_1}(\tilde x)'\ =\ \gd_{i,i_1}(\frac1{n_0}x'\tilde x\ +\ 1)\ =\ \gd_{i,i_1}(\Sigma^{(1)}(x,\tilde x)+1)\ =\ \gd_{i,i_1}\Theta^{(1)}(x,\tilde x)\,. 
\end{equation}
We now proceed by induction. Suppose we have proved the claim for $z^{(l-1)}$. We have 
\begin{equation}
z^{(l)}_i(x)\ =\  \frac1{n_l^{1/2}}\sum_{j=1}^{n_l}W_{i,j}^{(l)}\phi(z_j^{(l-1)}(x))+b^{(l)}
\end{equation}
The vector of coefficients of this neural network, $\theta,$ can be decomposed into $\theta=(\theta_{l-1}, W^{(l)},b^{(l)}),$ where $\theta_{l-1}$ are all the coefficients except for those of the $l$'th layer. Each $z_j^{(l)}(x),\ j=1,\cdots,n_{l}$. We have  
\begin{equation}
\begin{aligned}
&\nabla_{\theta_{l-1}}z^{(l)}_i(x)\ =\ \frac1{n_{l}^{1/2}}\sum_{j=1}^{n_l}W_{i,j}^{(l)}\nabla_{\theta_{l-1}}\phi(z_j^{(l-1)}(x))\\
&=\ \frac1{n_{l}^{1/2}}\sum_{j=1}^{n_{l}}W_{i,j}^{(l)}\nabla_{\theta_{l-1}}z_j^{(l-1)}(x)\phi'(z_j^{(l-1)}(x))
\end{aligned}
\end{equation}
and, hence, 
\begin{equation}
\begin{aligned}
&\nabla_{\theta_{l-1}}z^{(l)}_{i_1}(x)\nabla_{\theta_{l-1}}z^{(l)}_{i_2}(\tilde x)'\\ 
&=\ \frac1{n_{l}}\sum_{j_1=1}^{n_{l}}W_{i_1,j_1}^{(L+1)}\nabla_{\theta_{l-1}}z_{j_1}^{(l-1)}(x)\phi'(z_{j_1}^{(l-1)}(x))\sum_{j_2=1}^{n_{l}}W_{i_2,j_2}^{(l)}\nabla_{\theta_{l-1}}z_{j_2}^{(l-1)}(\tilde x)\phi'(z_{j_2}^{(l-1)}(\tilde x))\,.
\end{aligned}
\end{equation}
By the induction hypothesis, in the limit as $n_{l-1},\cdots,n_1\to\infty,$ we have 
\begin{equation}
\begin{aligned}
&\nabla_{\theta_{l-1}}z^{(l)}_{i_1}(x)\nabla_{\theta_{l-1}}z^{(l)}_{i_2}(\tilde x)'\\ 
&\approx\ \frac1{n_{l}}\sum_{j_1=1}^{n_{l}}W_{i_1,j_1}^{(L+1)}\phi'(z_{j_1}^{(l-1)}(x))\sum_{j_2=1}^{n_{l}}W_{i_2,j_2}^{(l)}\phi'(z_{j_2}^{(l-1)}(\tilde x))\gd_{j_1,j_2}\Theta^{(l)}(x,\tilde x)\\
&=\ \frac1{n_{l}}\sum_{j=1}^{n_{l}}W_{i_1,j}^{(L+1)}\phi'(z_{j}^{(l-1)}(x))W_{i_2,j}^{(l)}\phi'(z_{j}^{(l-1)}(\tilde x))\Theta^{(l)}(x,\tilde x)
\end{aligned}
\end{equation}
In the limit as $n_l\to\infty,$ we get 
\begin{equation}
\begin{aligned}
&\nabla_{\theta_{l-1}}z^{(l)}_{i_1}(x)\nabla_{\theta_{l-1}}z^{(l)}_{i_2}(\tilde x)'\\ 
&\to\gd_{i_1,i_2}(\gs_W^{(l)})^2E[\phi'(z_{j}^{(l-1)}(x))\phi'(z_{j}^{(l-1)}(\tilde x))]\Theta^{(l)}(x,\tilde x)
\end{aligned}
\end{equation}
At the same time, 
\begin{equation}
\begin{aligned}
&\nabla_{(W^{(l)},b^{(l)})}z^{(l)}_{i_1}(x)\nabla_{(W^{(l)},b^{(l)})}z^{(l)}_{i_2}(\tilde x)'\\ 
&=\gd_{i_1,i_2}(\frac1{n_l}\sum_j \phi(z_{j}^{(l-1)}(x))\phi(z_{j}^{(l-1)}(\tilde x))+1)\\
&\to\ \gd_{i_1,i_2}(E[\phi(z_{j}^{(l-1)}(x))\phi(z_{j}^{(l-1)}(\tilde x))]+1)\\
&=\ \gd_{i_1,i_2}(\Sigma^{(l)}(x,\tilde x)+1)
\end{aligned}
\end{equation}
\end{proof}

\section{Experimental Setup}
\label{sec:experiments}
This appendix describes the full empirical pipeline. We first present the MLP training configuration (\ref{subsec:mlp_training}), then detail the construction of PTK and random feature factors (\ref{subsec:factor_construction}), followed by performance evaluation metrics and statistical diagnostics (\ref{subsec:random_seeds}), and finally robustness procedures based on multiple random seeds and ensembling (\ref{subsec:stats}).

\subsection{MLP Training Configuration}\label{subsec:mlp_training}

The MLP models are trained using a rolling-window scheme designed to reflect realistic real-time forecasting and portfolio construction. At each month $t$, the model is estimated using data from the previous 60 months and then used to generate portfolio weights for month $t+1$. The training procedure, architecture, and optimization details are summarized below.

\paragraph{Model architecture.} 
We employ a standard parametrization MLP (SP-MLP) with ReLU activations. Depths: 1, 2, 4 layers; widths: 16–512 neurons (16, 32, 64, 128, 256, 512). Total configurations: 18 (3×6).

\paragraph{Initialization.} Network weights are initialized using the default PyTorch Kaiming uniform scheme, where weights are drawn from $\mathcal{U}(-\sqrt{k}, \sqrt{k})$ with $k = 1/\text{in\_features}$, and biases are initialized to zero.

\paragraph{Training scheme.} For the initial training window, the model is trained from scratch for 20 epochs. Thereafter, as the training window rolls forward by one month (dropping the oldest observation and adding the newest), the model is warm-started from the previous solution and fine-tuned for 10 epochs. This rolling fine-tuning scheme substantially reduces computational cost while preserving stability and convergence.

\paragraph{Mini-batching and optimization.} Optimization is performed using the Adam optimizer, with learning rate $\text{lr} = 2^{-\text{lr\_power}}$ and $\text{lr\_power} = 16$. Training is conducted with a batch size of one, where each batch corresponds to the entire cross-section of stocks in a given month. Specifically, at each training step, the model simultaneously processes all $N_t$ stocks observed in month $t$, producing a single portfolio return forecast. This design aligns the training objective directly with the portfolio construction problem.

\paragraph{Loss function and normalization.} The model is trained to minimize the maximal Sharpe ratio regression (MSRR) loss. To ensure scale invariance and stable optimization across months with differing cross-sectional sizes, the loss is normalized by the number of stocks in the cross-section. Formally, letting $f_t \in \mathbb{R}^{N_t}$ denote the vector of MLP predictions and $R_{t+1} \in \mathbb{R}^{N_t}$ the realized excess returns, the loss is defined as
\[
\mathcal{L}_t
\;=\;
\frac{\bigl(1 - f_t^\top R_{t+1}\bigr)^2}{N_t}.
\]
This normalization ensures that each monthly training step contributes comparably to the overall optimization objective, regardless of fluctuations in the cross-sectional sample size.

\paragraph{Training window.} Throughout, we employ a rolling training window of length 60 months.

\subsection{Factor Construction}\label{subsec:factor_construction}

\subsubsection{Portfolio Tangent Kernel (PTK) Factors}

The PTK factors are derived from the gradients of the trained MLP with respect to its parameters. For a given month $t$ with $N_t$ stocks, let $f(X_{i,t}; \theta_t^*)$ denote the MLP prediction for stock $i$.

\begin{itemize}

\item \textbf{Gradient Computation}: We compute the gradient of the normalized factor return (MSRR loss) with respect to the network parameters:

\[
g_t(\theta_t^*) 
=
\nabla_\theta \left(
\frac{1}{\sqrt{N_t}} \sum_{i=1}^{N_t} f(X_{i,t}; \theta_t^*) R_{i,t+1}
\right) \in \mathbb{R}^{1\times P},
\]

where $R_{i,t+1}$ is the realized excess return and $P$ is the total number of network parameters. Stacking these vectors over time produces the gradient matrix

\[
G = 
\begin{bmatrix}
g_1 \\
g_2 \\
\vdots \\
g_T
\end{bmatrix} \in \mathbb{R}^{T \times P}.
\]

\item \textbf{Kernel Matrix}: The kernel matrix $\bK$ is computed as the uncentered Gram matrix of the gradients:
\[
\bK = \frac{1}{T} G G' \in \mathbb{R}^{T \times T}.
\]

\item \textbf{Ridge Regression}: We solve for the portfolio weights using ridge regression on the kernel matrix with penalty parameter $z$.

\item \textbf{Hyperparameters}:
\begin{itemize}
\item \textbf{Ridge penalties ($z_{\mathrm{eff}}$)}: 8 values logarithmically spaced from $10^{-5}$ to $100$.
\item \textbf{Rolling windows}: $\{12, 60, 120, 240, 360\}$ months.
\end{itemize}

\end{itemize}

\subsubsection{ReLU (random-feature) Factors}
To provide a baseline for the PTK, we generate random-feature factors that approximate infinite-width (dimension $P=25,000$).

\begin{itemize}
\item \textbf{Weight Initialization}: Random weights $\Omega \in \mathbb{R}^{P \times d}$ are sampled from a scaled Gaussian distribution:
\[
\Omega_{ij} \sim \mathcal{N}\left(0, \frac{2\sigma^2}{d}\right)
\]
where $d$ is the input dimension and $\sigma$ is a scaling parameter.

\item \textbf{Activations}: \textbf{ReLU}: $\phi(x) = \max(0, x \Omega^T)$.

\item \textbf{Factor Construction}: The random-feature factors $F_t$ are constructed as managed portfolios, scaled by $N_t^{-1/2}$:
\[
F_{j,t} = \frac{1}{\sqrt{N_t}} \sum_{i=1}^{N_t} \phi_j(x_{i,t}) R_{i,t+1}
\]
\end{itemize}

\subsection{Random Seeds and Robustness}\label{subsec:random_seeds}
To ensure the robustness of our results and account for the variability in initialization:
\begin{itemize}
\item \textbf{Random Seeds}: All experiments are repeated across multiple random seeds (10 seeds).
\begin{itemize}
\item For \textbf{MLP} models, the seed controls the initialization of the neural network weights.
\item For \textbf{random-feature} models, the seed determines the sampling of the random weight matrix $\Omega$.
\end{itemize}
\item \textbf{Ensembling}: In our analysis, we often report the performance of an ensemble of models trained with different seeds. The ensemble return is computed as the volatility-weighted average of the individual seed returns:
\[
R_{ensemble, t} = \sum_{k=1}^{K} w_{k,t} R_{k,t}
\]
where $w_{k,t} \propto 1/\widehat\sigma_{k,t}$ is the inverse rolling volatility of the $k$-th seed's strategy.
\end{itemize}

\subsection{Performance Evaluation and Statistical Tests} \label{subsec:stats}
We evaluate the strategies using the following metrics:

\begin{itemize} 
\item \textbf{Alpha ($\alpha$)}: The intercept from the OLS regression of the strategy's excess returns ($R_{strat}$) on the benchmark's excess returns ($R_{bench}$):
\[
R_{strat, t} = \alpha + \beta R_{bench, t} + \epsilon_t\,,
\]
with Newey–West (1987) standard errors using appropriate lag length.
\item \textbf{Sharpe Ratio}: The annualized risk-adjusted return, calculated as $\sqrt{12} \cdot \frac{\mu}{\sigma}$, where $\mu$ and $\sigma$ are the mean and standard deviation of OOS monthly returns.
\end{itemize}

\subsubsection{Implicit Regularization and Debiased GRS}
To characterize model complexity, overfitting, and generalization behavior, we analyze the spectral properties of the kernel matrix using the following statistics:

\begin{itemize}
\item \textbf{Implicit Regularization ($Z_*(z)$)}: The effective degrees of freedom of the kernel matrix $K$ with ridge penalty $z$:
\[
Z_*(z) = \frac{T}{\text{tr}((zI + K)^{-1})} = \frac{T}{\sum_{i=1}^T \frac{1}{z + \lambda_i}}
\]
where $\lambda_i$ are the eigenvalues of $K$ and $T$ is the window size. We often report the normalized metric divided by $\frac{1}{T}\text{tr}(K)$.

\item \textbf{Debiased GRS Statistic ($\hat W_{\text{debiased}}(z)$)}: To correct for the bias in the standard GRS statistic when $P/T$ is large, we compute:
\[
\hat W_{\text{debiased}}(z)
\]
using formula \eqref{debiased}.
\end{itemize}

\section{Additional Plots}

\begin{figure}[htbp]
\centering
\includegraphics[width=0.8\textwidth]{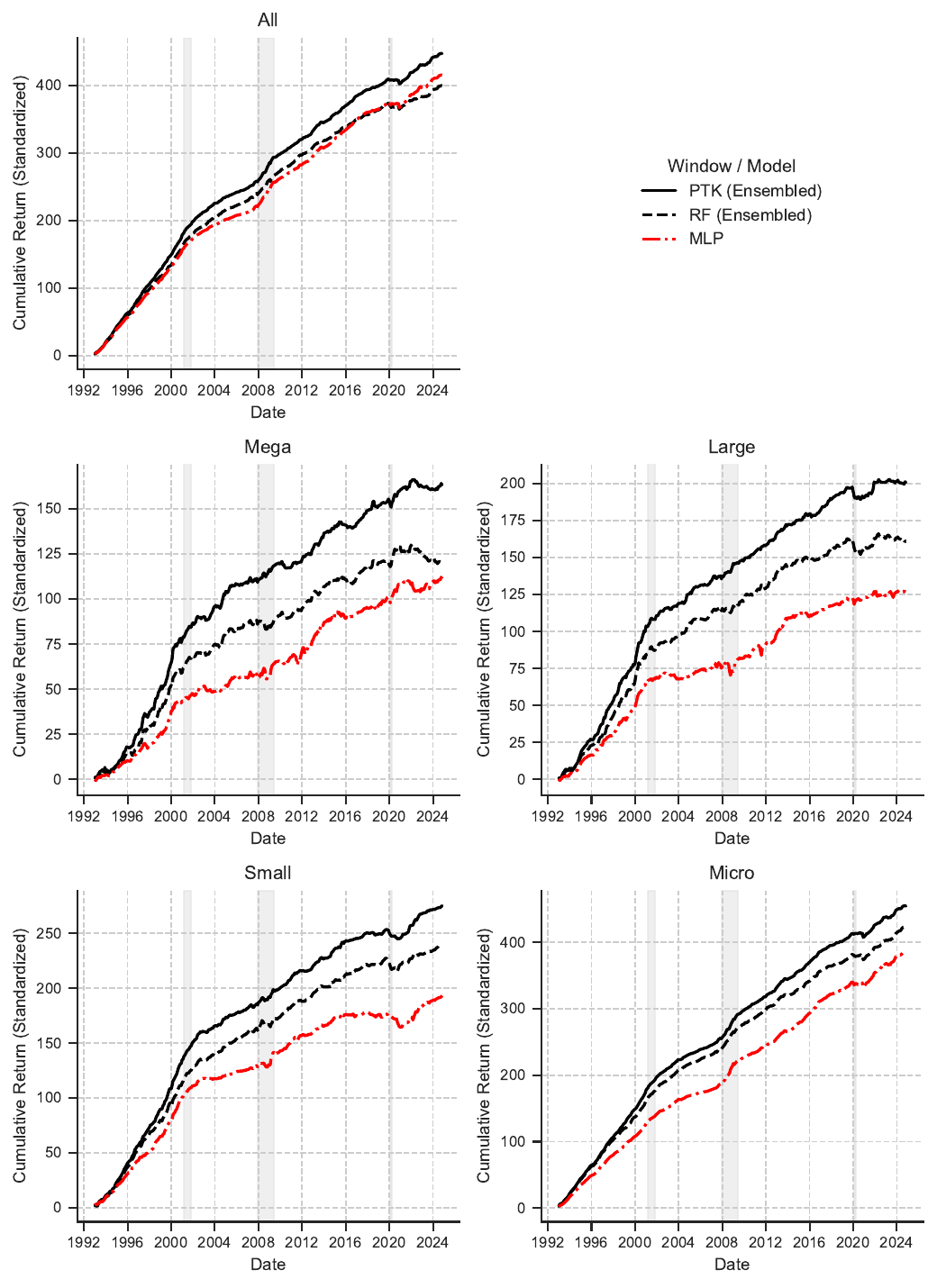}
\caption{Cumulative Returns of $\bar R^{\rm PTK}$ from \eqref{the-ptk-bench}, $\bar R^{\textrm{ReLU}}$ from \eqref{the-ReLU-bench}, and $R^{\rm DNN}$ from \eqref{def:rdnn}. We use $D=2,\ w=256$ and $z_{eff}=10^{-5}.$ All returns are normalized to equal full-sample variance.}
\label{fig:ptk_cumulative_returns}
\end{figure}

\begin{figure}[htbp]
\centering
\includegraphics[width=0.8\textwidth]{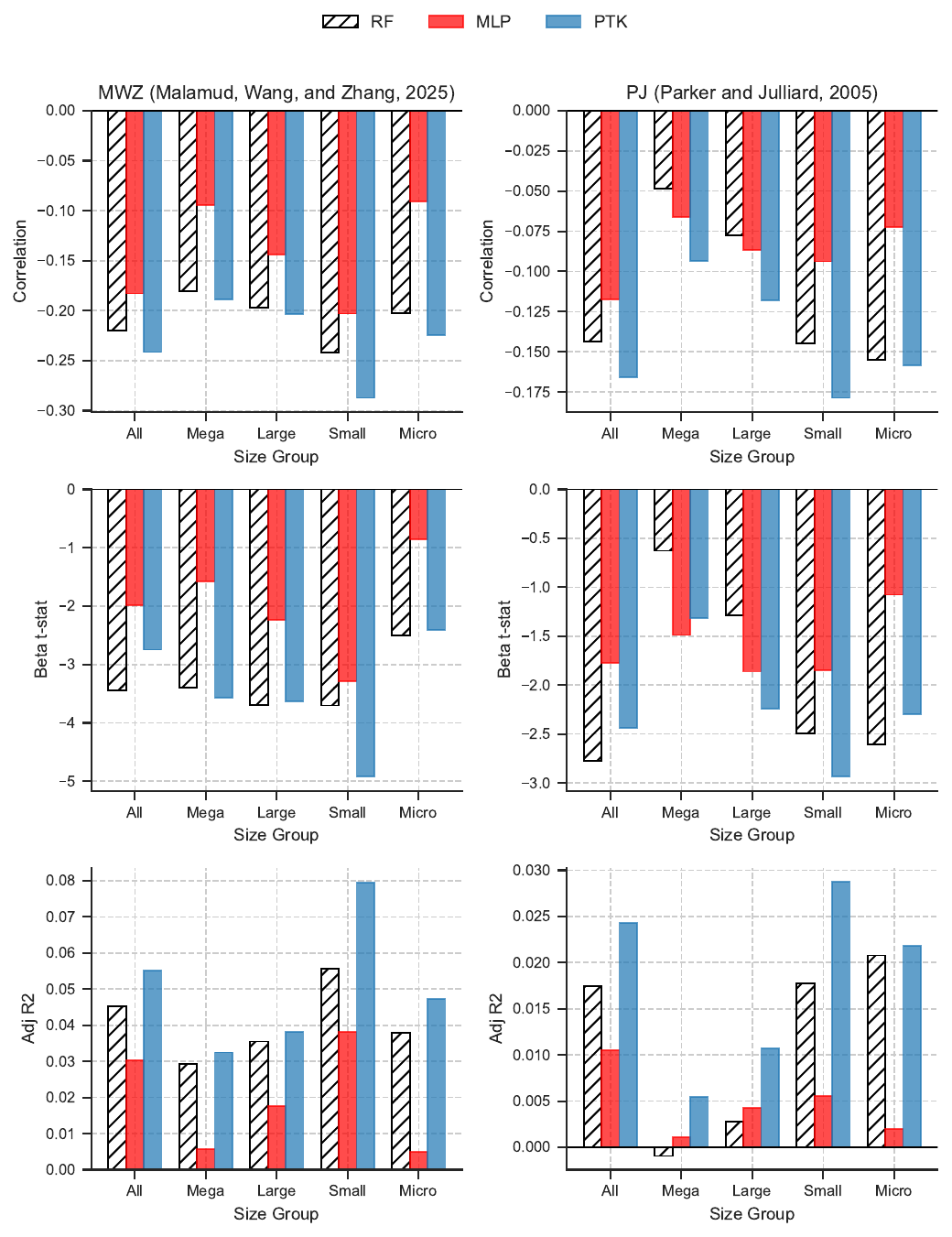}
\caption{Machine Learning SDFs and Consumption Risk.
This figure reports the relation between tradeable machine-learning stochastic discount factors (SDFs) and consumption-based SDFs from \cite{ParkerJulliard2005} (PJ) and \cite{MalamudWangZhang2025} (MWZ). The top row shows time-series correlations between portfolio returns from the fully trained DNN SDF $R^{\mathrm{DNN}}$ \eqref{def:rdnn} (red), the random-feature benchmark $\bar R^{\mathrm{ReLU}}$ \eqref{the-ReLU-bench} (grey, striped), and the PTK-based SDF $\bar R^{\mathrm{PTK}}$ \eqref{the-ptk-bench} (blue), and the MWZ future-consumption SDF (left panel) and the PJ long-run consumption SDF (right panel). 
The middle row reports $t$-statistics from univariate time-series regressions of each consumption-based SDF on the corresponding tradeable SDF return. The bottom row reports the associated adjusted $R^2$. All statistics are shown separately by size group. Consumption-based SDFs are constructed using future realized consumption growth aggregated over horizons of 84 months (MWZ) and 36 months (PJ), respectively, which truncates the sample in June 2018. Ridge regularization with $z_{eff}=10^{-5}$ is applied to the ReLU and PTK portfolios.}
\label{fig:links_to_ccapm}
\end{figure}

\end{document}